\documentclass[12pt, draftclsnofoot, onecolumn]{IEEEtran}
\usepackage{etoolbox}
\usepackage{amssymb, cite}
\usepackage{bm,bbm}
\usepackage{mathtools}
\usepackage{multirow,color,amsfonts}
\usepackage{tabulary}
\usepackage{subfigure}
\usepackage{graphicx}
\usepackage{setspace}
\usepackage{enumerate}
\usepackage[]{algorithm2e}
\usepackage{comment}
\usepackage[utf8]{inputenc} 
\usepackage[T1]{fontenc}
\usepackage{url}
\usepackage{ifthen}
\usepackage{cite}
\usepackage{amsthm}
\usepackage{caption}
\usepackage{enumitem}

\newtheorem{theo}{Theorem}
\newtheorem{defi}{Definition}

\newtheorem{prop}{Proposition}
\newtheorem{remark}{Remark}

\newtheorem{ex}{Example}

\DeclareMathOperator\erfc{erfc}

\newcommand{\Kgkw}{K}

\newcommand\norm[1]{\left\lVert#1\right\rVert}
\newcommand\normtwo[1]{{\left\lVert#1\right\rVert}_2}

\newcommand\independent{\protect\mathpalette{\protect\independenT}{\perp}}
\newcommand*\xor{\mathbin{\oplus}}
\def\independenT#1#2{\mathrel{\rlap{$#1#2$}\mkern2mu{#1#2}}}

\IEEEoverridecommandlockouts
\begin{document}

\title{A Distributed Computationally Aware Quantizer Design via Hyper Binning} 

\author{Derya~Malak and Muriel M\'{e}dard \thanks{D. Malak is with the Communication Systems Dept., EURECOM, Biot Sophia Antipolis, FRANCE (derya.malak@eurecom.fr).} 
\thanks{M. M\'{e}dard is with RLE, MIT, Cambridge, MA, USA (medard@mit.edu).}
\thanks{An early version of the paper appeared in Proc. IEEE SPAWC 2020 \cite{malakmedard2020}. \hfill Manuscript last revised: {\today}.} }

\maketitle

\begin{abstract}
We design a distributed function-aware quantization scheme for distributed functional compression. 
We consider $2$ correlated sources $X_1$ and $X_2$ and a destination that seeks an estimate $\hat{f}$ for the outcome of a continuous function $f(X_1,\,X_2)$. We develop a compression scheme called hyper binning in order to quantize $f$ via minimizing the entropy of joint source partitioning.  Hyper binning is a natural generalization of Cover's random code construction for the asymptotically optimal Slepian-Wolf encoding scheme that makes use of orthogonal binning. The key idea behind this approach is to use linear discriminant analysis in order to characterize different source feature combinations. This scheme captures the correlation between the sources and the function's structure as a means of dimensionality reduction. We investigate the performance of hyper binning for different source distributions and identify which classes of sources entail more partitioning to achieve better function approximation.
Our approach brings an information theory perspective to the traditional vector quantization technique from signal processing.
\end{abstract}

\section{Introduction}
\label{intro}

Compression and processing of large amount of data is a challenge in various applications. From an information theory perspective, there are asymptotic optimal approaches to the distributed source compression problem that can achieve arbitrarily small decoding error probability for 
large blocklengths, such as noiseless distributed coding of correlated sources as proposed by Slepian-Wolf \cite{slepian1973noiseless}, and their extensions \cite{wyner1976rate,PradRam2003,coleman2006low}, which are based on orthogonal binning of typical sequences. Practical Slepian-Wolf encoding schemes include coset codes \cite{PradRam2003}, trellis codes \cite{wang2001design}, 
and turbo codes \cite{bajcsy2001coding}. Other examples include rate region characterization using 
a graph-based approach, 
such as a graph-based approach, 
such as \cite{korner1973coding,AlonOrlit1996,OrlRoc2001,DosShaMedEff2010,feizi2014network,FES04,Gal88}, and coding for computation with communication constraints \cite{LiAliYuAves2018,
YuAliAves2018}. While some approaches focus on network coding for computing linear functions, such as \cite{KowKum2012,HuanTanYangGua2018,AppusFran2014,koetter2003algebraic,LiYeuCai2003,HoMedKoeKarEffShiLeo2006}, there exist works exploiting functions with special structures, e.g., in \cite{SheSutTri2018} as well as coding 
of sparse graphical data, e.g.,  
\cite{delgosha2019notion} and  \cite{delgosha2018distributed}.

The related work in the signal processing domain includes vector quantization and distributed estimation-based models. A vector quantization technique was proposed in \cite{padmanabhan1999partitioning}, where the feature space is partitioned via a hierarchical tree-based classifier such that the average entropy of the class distribution in the partitioned regions is minimized. In \cite{ribeiro2006bandwidthI}, conditions for efficiently quantizing scalar parameters were characterized and 
estimators that require transmitting just one bit per source that exhibits variance almost equal to the minimum variance estimator based on unquantized observations were proposed. Max-Lloyd algorithm, which is a Voronoi iteration method, was applied to vector quantization and pulse-code modulation \cite{max1960quantizing}. 
Vector quantization using linear hyperplanes was applied to distributed estimation in sensor networks in the presence of noise \cite{fang2009hyperplane}, and with resource constraints \cite{ribeiro2006bandwidthII}. In addition to the quantization-based approaches, the problem of detection and hypothesis testing have drawn significant attention, see the schemes, e.g., a mismatched detector for channel coding and hypothesis testing \cite{abbe2007finding}, or signal constellation design with maximal error exponent \cite{huang2004error}. There has recently been quite interesting work in traditional signal processing that minimizes some distortion measure from the quantized measurements, e.g., hardware-limited quantization for achieving the minimum mean-squared error (MMSE) distortion \cite{shlezinger2019hardware}, 
task-based quantization for recovering functions with special structures, e.g., quadratic functions as in \cite{salamatian2019task}, and sparse functions \cite{cohen2019serial2}.

Another perspective on efficient representation is coding for functional compression, which is complementary to the vector quantization methods. In \cite{basu2020hypergraph}, the authors have proposed a hypergraph-based coloring scheme whose rate lies between the Berger-Tung inner and outer bound and showed that for independent sources, their scheme is optimal for general functions. In \cite{servetto2005achievable}, the author has derived inner and outer bounds for multiterminal source coding. The author has shown that for scalar codes (scalar quantizers followed by block entropy coders) the two bounds converge.  
In \cite{misra2011distributed}, the authors have considered the distributed functional source coding problem, in which the sink node computes an estimate of the function $g(X_1,\dots, X_s)$ under MSE distortion. The setting is restricted to the communication of source data over rate-limited links, and scalar quantization of each $X_i$ for $i=1,\dots,s$ using a sequence of companding quantizers $\{Q_K^i\}$ of increasing resolution $K$, mostly for independent sources $\{X_i\}_{i=1}^s$.  
Unlike \cite{misra2011distributed}, we consider vector quantization without the assumptions on the source independence or the rate-limited links. 

\begin{table*}[t!]\small
\captionsetup{font=small}
\setlength{\extrarowheight}{0.8pt}
\begin{center}
\begin{tabular}{l | l | l}
{\bf Problem types} & {\bf Side information} & {\bf Distributed source coding for computing}\\ 
\hline
$f(X_1,\,X_2)=(X_1,\,X_2)$ & Wyner and Ziv \cite{wyner1976rate} & Coleman et al. \cite{coleman2006low}, 
\\
& & Berger et al. \cite{berger1979upper},
\\
& & Barros and Servetto \cite{barros2003rate}, 
\\
& & Wagner et al. \cite{wagner2008rate}\\
\hline
Coding for computing general $f(X_1,\,X_2)$ & Yamamoto \cite{yamamoto1982wyner},  
& Feizi and M\'edard \cite{feizi2014network},\\ 
& Feng et al. \cite{FES04}, & Basu et al. \cite{basu2020hypergraph}\\
& Doshi et al. \cite{DosShaMedEff2010},  & \\ 
& Basu et al. \cite{basu2020functional} &\\
\hline
Product of two broadcast channels & Watanabe \cite{watanabe2013rate} & \\ 
Multiple access channel (MAC) & Rajesh et al. \cite{rajesh2008distributed} & Nazer and Gastpar \cite{nazer2007computation} 
\\
Two-hop and diamond networks & & Guo \cite{gu2009achievable} 
\end{tabular}
\end{center}
\caption{Research progress on nonzero-distortion source coding problems.}
\label{table:research}
\end{table*}

In \cite{linder1999high}, the authors have considered high-resolution source coding with multidimensional companding for non-difference distortion measures. In \cite{bucklew1984multidimensional}, the author has minimized the MSE for the Wyner-Ziv problem with decoder side information and functional distortion. In \cite{thao1996lower}, the authors have used a structured hyperplane wave partition model using a frame model -- a redundant set of basis vectors -- that provides $O(1/R^2)$ MSE distortion as a function of the redundancy $R$ \cite{thao1996lower}, and the follow-on works such as \cite{goyal1998quantized} have focused on deterministic qualities of quantization, and \cite{goyal2001quantized}, which concerns applying frames to a packet erasure network. This model, similar to network coding \cite{koetter2003algebraic}, serves for recovering the DoFs more effectively. Different from \cite{linder1999high,bucklew1984multidimensional,thao1996lower,goyal1998quantized,goyal2001quantized}, we assume a randomized model where reconstruction is not consistent.

The broad and common objective in these models is finding ways of effective compression and communication of massive data. This goal is realizable by capturing underlying redundancy both in data and functions, and recovering a sparse representation, or labeling, at the destination. From a practical perspective, the redundancy across geographically dispersed sources' data plays a big role and can provide significant gains in compression. Hence, from a technical point of view, compressing data is preferred for reducing resource consumption in networks (e.g., wireless or data center networks).  
Furthermore, there might be privacy concerns at the source sites because sources may not be willing to share  
sensitive data, including customer data or medical records. Additionally, the destination might only be interested in a function of the data and cannot store the entire data. In this scenario, the sources aim to collectively determine a function outcome without disclosing their data to each other.  
Hence, the distributed computation of functions naturally fits into the distributed source compression framework, ensuring the protection of sources.

We summarize the efforts on nonzero-distortion source encoding problems in Table \ref{table:research}. 
Despite these approaches, the exact achievable rate region for the function compression problem is, in general, an open problem. To the best of our knowledge, it is only solved for special scenarios, including general tree networks \cite{feizi2014network}, linear functions \cite{koetter2003algebraic}, 
identity function \cite{slepian1973noiseless}, and rate-distortion characterization with decoder side information \cite{wyner1976rate}. However, there do not exist tractable approaches that approximate the information-theoretic limits to perform functional compression in general topologies. Thus, unlike compression, for which coding techniques exist, and compressed sensing acts in effect as an alternative for coding, for purposes of simplicity and robustness, there is currently no family of coding techniques for functional compression.

Our main contributions are summarized as follows:
\begin{itemize}[leftmargin=*]
\item A novel approach, called \emph{hyper binning}, for distributed function-aware quantization that uses hyperplane arrangements (Sect. \ref{what_binning}). It provides a vector quantized functional representation of distributed sources that minimizes the entropy of joint source partitioning (Sect. \ref{problem}).
\item Application of hyper binning to sources modeled as a Gaussian mixture model (GMM) (Sect. \ref{HyperBinDesign}), as a special tractable case of the problem. To demonstrate the gains of hyper binning, we also consider more general continuous and discrete-valued sources (Sects. \ref{binning_distributed_source_coding}, \ref{comparison_existing_work}, and \ref{discussion}).  
\item The theoretical justification for the rate-distortion performance of \emph{hyper binning}, for distortion criteria including a) entropy-based, b) mean-squared error (MSE), or c) Hamming distortion and d) Gaussian approximation. (Sect. \ref{distortion}, where we provide the rate-distortion expressions for the general case and the case of the GMM).
\item A scaling between the number of hyperplanes $J$ and the blocklength $n$ that {\em hyper binning} can support (Sect. \ref{convex_background}).  
\item Characterization of the description length of hyper binning at finite blocklengths via 
Kolmogorov complexity (Sect. \ref{compression_finite_blocklength}).
\item A comparison of hyper binning for real-valued source data with coloring-based modular compression schemes that decouple quantization and binning (orthogonal binning, e.g., Slepian-Wolf coding \cite{slepian1973noiseless}, or codebook trimming \cite{feizi2014network} [Sect. \ref{binning_distributed_source_coding}]) for discrete-valued data (Sect. \ref{comparison_existing_work}).
\item An encoding heuristic for hyper binning by exploiting the G{\'a}cs-K{\"o}rner common information (GK-CI) (Sect. \ref{discussion}).
\end{itemize} 
Via the proposed hyper binning scheme, we aim to address the following central questions:
\begin{itemize}
    \item What functions $f$ can we approximate well? Via hyper binning, the class of functions $f$ we can compute (with zero error) is the class of $f$ given by a hyperplane arrangement. Our approach can be used to well approximate discrete, piecewise constant, or linearly separable functions. In addition, hyper binning can model $f$ that are continuous in the neighborhood of the quantization levels.  
    For other classes of $f$, the approximation depends on the distortion metric and criterion and the number of hyperplanes in general position (GP), which divides the space into a maximum number of regions \cite{chataignon2019comparison}.
    \item How should we choose the hyperplane parameters given a function $f$? 
    This choice depends on the distortion criterion. For instance, in Sect. \ref{HyperBinDesign}, we detail source data satisfying a Gaussian mixture model subject to a given LDA classification error criterion. However, for general source distribution models subject to different distortion metrics, $f$ significantly impacts the design. To that end, we consider various examples in Sect. \ref{comparison_existing_work}.
    \item How many hyperplanes $J$ do we require at finite blocklengths? How does $J$ scale as $n\to\infty$? 
    The maximum blocklength that can be supported with $J$ hyperplanes in GP is  $n_{\max}=\frac{J}{2}+O\left(\frac{1}{J}\right)$, i.e., $J\approx 2n_{\max}$ as $n\to\infty$.  
\end{itemize} 

\begin{figure*}[h!]
\centering
\includegraphics[width=0.85\textwidth]{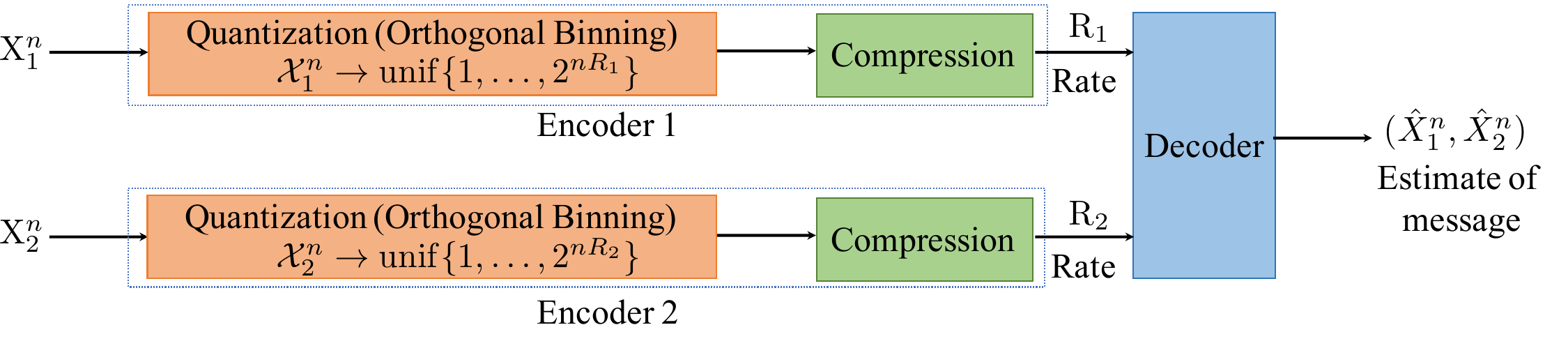}
\caption{\small{Distributed compression scheme of Slepian-Wolf \cite{slepian1973noiseless} via binning constructed using the asymptotically optimal approach in \cite{cover2012elements}. }}\label{fig:SW}
\end{figure*}

\begin{figure*}[h!]
\centering
\includegraphics[width=0.85\textwidth]{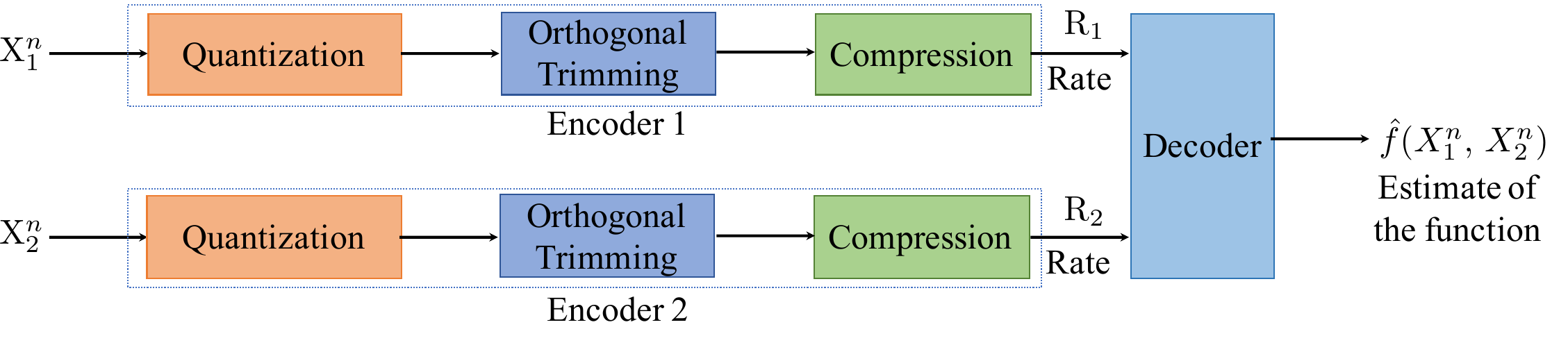}
\caption{\small{Orthogonal trimming of the random binning-based  
codebook, where the bins of $X_1^n$ and $X_2^n$ are independently trimmed.}}\label{fig:orthogonal_binning}
\end{figure*}

{\bf Connections to the State-of-the-Art.} The novelty of {\em hyper binning} is that we find a partitioning of the sources using a hyperplane arrangement to allow describing a function of interest up to some quantization distortion rather than determining an independently quantized representation of dispersed source data oblivious to the function. Such a scheme needs fewer dimensions than the codeword size and captures the function's dependence on the data. 
The technique differs from traditional vector quantization for data compression and brings together techniques from information theory, such as distributed source encoding, functional compression, and optimization of mutual information, to the area of signal processing via function quantization inspired by hyperplane-based vector quantizers.  
Hyper binning does not rely on the NP-hard nature of graph coloring \cite{feizi2014network} and the asymptotically optimal information-theory-based models \cite{slepian1973noiseless}, \cite{wyner1976rate} which are impractical for finite blocklengths. Hyper binning is an intuitive generalization using linear hyperplanes for encoding continuous functions through a vector quantization of the high dimensional codebook space. 
Our results can be used to recover the instances, e.g., the Slepian-Wolf compression model or its orthogonal trimming.

{\bf Organization.} The paper's organization to answer the central questions outlined above is as follows. Sect. \ref{problem} states the problem of vector quantized functional representation of distributed sources and describes a linear hyperplane-based distributed function encoding approach called {\em hyper binning}. Sect. \ref{why_binning} provides the motivation behind, Sect. \ref{convex_background} details background on convex sets and hyperplanes, and Sect. \ref{NecessaryConditions} details the necessary conditions for encoding the functions.  
Sect. \ref{HyperBinDesign} focuses on the analytical details of hyper binning for encoding functions to determine the optimal hyperplane allocation for a specific instance where the source data is characterized by a Gaussian mixture model (GMM). More specifically, for the GMM, Sect. \ref{data} describes the data and hyperplane arrangement, Sect. \ref{HyperplaneDistribution} focuses on optimizing the arrangement to maximize a notion of the mutual information between the function and the partitions, Sect. \ref{source_data_distribution_symmetric_asymmetric} describes the behavior of the mutual information for different source data distribution models across the  
classes of the GMM and Sect. \ref{distortion} details several rate-distortion models (including the entropy-based, mean-squared, Hamming, and Gaussian distortion models) of hyper binning for characterizing the GMM. 
Sect. \ref{binning_distributed_source_coding} contrasts hyper binning and orthogonal binnings for infinite blocklengths, along with the assumptions on the sources, via building on the classical distributed encoding approach of \cite{slepian1973noiseless}. 
Sect. \ref{compression_finite_blocklength} is concerned with the compression complexity at finite blocklengths. 
To demonstrate the gains of hyper binning for more general source distributions, including continuous-valued sources, Sect. \ref{comparison_existing_work} provides a rate-region comparison of hyper binning and existing schemes on graph-based \cite{feizi2014network} and hypergraph-based \cite{basu2020hypergraph,basu2020functional} coloring schemes, both for pre and post-quantized source data. 
Sect. \ref{discussion} details a discussion on the connections between hyper binning and coloring-based coding models and a heuristic for encoding that relies on the G{\'a}cs-K{\"o}rner CI. Finally, Sect. \ref{conclusion} summarizes our contributions and points out future directions.

{\bf Notation.}
The binary entropy function, denoted $h(p)$, satisfies $h(p) = -p\log _{2}p-(1-p)\log _{2}(1-p)$. Given a discrete random variable $X$, $H(X)=\mathbb{E}[-\log_2(X)]$ is the entropy of $X$ in bits. Similarly, $H(X_1,X_2)$ is the joint entropy of $X_1$ and $X_2$, and $H(X_1|X_2)$ is entropy of $X_1$ conditioned on $X_2$. 

Let $C$ be a non-empty closed convex subset of $\mathbb{R}^n$, i.e., $C \subseteq\mathbb{R}^n$, and ${\bf x},\,{\bf z}$ be vectors in $\mathbb{R}^n$, and $\norm{\cdot}$ denote the Euclidean norm on $\mathbb{R}^n$. For $n \in N$, let $B^n= \{{\bf x} \in \mathbb{R}^n: \norm{{\bf x}} \leq 1\}$ be the unit ball, and $\nu_{n-1}$ denote the uniform distribution on the unit sphere $S^{n-1}= \{{\bf x} \in \mathbb{R}^n: \norm{{\bf x}} = 1\}$.

\section{Problem Statement}
\label{problem} 
We consider a system with two encoders (the problem can be generalized to any number of encoders $s>2$) and a joint decoder. For a given blocklength $n$, two encoders observe random sequences ${\mathbf{X}_1^n} \in \mathcal{X}_1^n$ and ${\mathbf{X}_2^n} \in \mathcal{X}_2^n$ where the pairs $\{ (X_{1}(l), X_{2}(l)) \colon l = 1, \ldots, n\}$ are two statistically dependent and length $n$ sequences drawn 
independently and identically (i.i.d.) according to a known joint distribution $p_{X_1 ,X_2}(x_1, x_2)$, i.e., the sequences have a joint distribution that satisfies $\prod\nolimits_{l=1}^n p_{X_1 ,X_2}(x_1(l), x_2(l))$ for ${\mathbf{x}_i^n}\in\mathcal{X}_i^n$, $i\in\{1,2\}$.   
The decoder aims to recover a vector quantized functional representation of distributed sources. The source terminals must independently encode these observations into messages sent to a decoder who wishes to estimate the sequence $f({\mathbf{X}_1^n},\,{\mathbf{X}_2^n})=\{ f(X_{1}(l), X_{2}(l)) \colon l = 1, \ldots, n\}$ subject to distortion (which we detail in Sect. \ref{distortion}), where $f : \mathcal{X}_1 \times \mathcal{X}_2 \to \mathcal{Y}$ is a single-letter function that could be continuous or discrete. In particular, we are interested in the case where $\mathcal{X}_1 = \mathcal{X}_2 = \mathcal{Y} \subseteq \mathbb{R}$,  
i.e., $X_1$ and $X_2$ could have bounded support and $f$ could be defined on a bounded subset of $\mathbb{R}^2$. We assume that $f$ is known both at the sources and the decoder, and there is no feedback in the system.

\begin{figure*}[h!]
\centering
\includegraphics[width=0.85\textwidth]{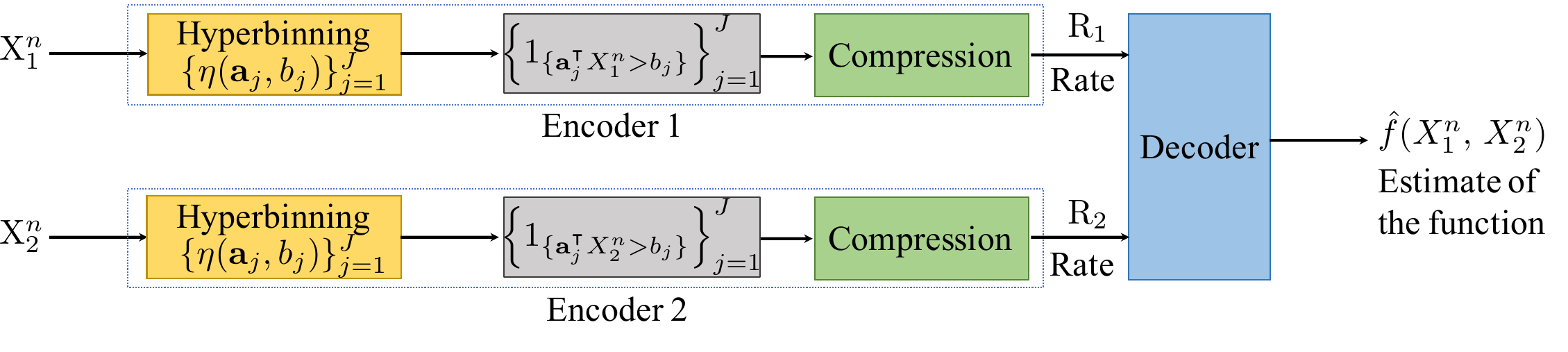}
\caption{\small{Hyper binning, which generalizes the orthogonal compression to convex regions determined by the intersections of hyperplanes.}}\label{fig:hyperbinning}
\end{figure*}

Some special cases of this distributed quantization problem have been considered in the literature, including the distributed encoding scenario studied by Slepian-Wolf in their landmark paper when $f$ is the identity function \cite{slepian1973noiseless}. Specifically, given sources $X_1$ and $X_2$ with finite alphabets, 
the Slepian-Wolf theorem gives a theoretical bound on the lossless compression rate for distributed coding of two statistically dependent and i.i.d. finite alphabet source sequences  \cite{slepian1973noiseless}. Indeed, Cover developed an asymptotically optimal encoding scheme using orthogonal binning \cite{cover2012elements}. The orthogonal binning is such that the codewords are selected uniformly at random from each bin, and the bins are equally likely. The functional extensions of \cite{slepian1973noiseless} in \cite{feizi2014network} and \cite{basu2020hypergraph} are also on the already post-quantized data streams.  
In this paper, we generalize the coding approach in \cite{slepian1973noiseless}. Instead of decoding the identity function, i.e., the sources $X_1$ and $X_2$ themselves, we recover a continuous function $f$ of $X_1$ and $X_2$ that satisfies the properties detailed in Sect. \ref{NecessaryConditions}.

A common assumption in the point-to-point model of Shannon \cite{shannon1948mathematical} or more general communication systems is that signal is discrete-time sequence $X(l)$, $l = 1, \ldots, n$. The goal is to design a distributed compression scheme that gives the best possible reconstruction for a given distortion criterion. 
In particular, a classical approach is scalar quantization of source data samples, which turns the source data into a discrete memoryless sequence, followed by distributed compression, allowing the use of distributed source coding techniques, as shown in Fig. \ref{fig:SW}. Another approach generalizes the above quantization scheme to a compression model for estimating a class of functions that allows orthogonal trimming of codebooks, as shown in Fig. \ref{fig:orthogonal_binning}. In this approach, the random bins (uniformly quantized bins) generated by the coding theorem of Slepian-Wolf \cite{slepian1973noiseless} are trimmed orthogonally, i.e., the trimming of the sequences $X_1^n$ and $X_2^n$ is independently performed. While the theorem of Slepian-Wolf is originally for discrete variables, the rate-distortion function of Wyner–Ziv coding is known for both discrete and continuous alphabet cases of the source and the side information with a general distortion metric in \cite{liu2006slepian}, \cite{wyner1976rate}, \cite{wyner1978rate}. The designs in Figs. \ref{fig:SW} and \ref{fig:orthogonal_binning} are optimal only for a set of functions (piecewise constant or block). However, the separation-based approach (which first quantizes and then compresses the data) may be suboptimal. By contrast, a strategy that employs compression on the functional representation of the vector quantized data can outperform  separation. In this paper, we go beyond the classical compression algorithms that work on the post-quantized single-letter representation of data. 
We propose a novel linear hyperplane-based function encoding approach, called \emph{hyper binning}, that can operate on the pre-quantized data, using ideas from vector quantization to provide a more effective way of functional compression.

\subsection{What is Hyper Binning?}\label{what_binning}
Hyper binning relies on quantizing ${\mathbf{X}_1^n}$ and ${\mathbf{X}_2^n}$ using a collection of linear hyperplanes called a \emph{hyperplane arrangement}. A linear hyperplane is an $(n-1)$-dimensional subspace of an n-dimensional vector space and hence can be described with a linear equation of the following form:
\[
a_{1}x_{1}+a_{2}x_{2}+\cdots +a_{n}x_{n}=b.\ 
\]
The idea is to partition a high dimensional codebook space into closed convex regions called hyper bins that capture the correlations between $X_1$ and $X_2$ as well as the dependency between the function $f$ and $(X_1,X_2)$. The key intuition is that closed convex sets have dual representations as an intersection of half-spaces. For this purpose, we use a finite set of hyperplanes, and their crossings determine the hyper bins, i.e., the quantized outcomes of $f$. Via hyper binning, it is possible to represent $f$ accurately up to a distortion level. The quantization error can vanish by optimizing the number, parameters, and dimensions of the hyperplanes employed. To the best of our knowledge, hyper binning is a new functional viewpoint to the challenging problem of distributed function-aware quantization in computational information theory. 

We denote a hyperplane arrangement with cardinality $J$ by $\{\eta({\bf a}_j, b_j)\}_{j=1}^J$. The choice of the hyperplane parameters $\{{\bf a}_j\in\mathbb{R}^n,\, b_j\in\mathbb{R}: j=1,\dots J\}$ depends on the characteristics of the joint distribution of $X_1$ and $X_2$ and its relation with the function $f({\mathbf{X}_1^n},\,{\mathbf{X}_2^n})$ to be estimated, which is detailed in Sect. \ref{HyperBinDesign}. 
In our encoding approach, unlike the orthogonal binning and orthogonal trimming approaches, we determine the half-spaces determined by the arrangement, where a half-space corresponding to hyperplane $j$ is a set given by $\{{\bf x} \in \mathbb{R}^n \colon {\bf a}_j^{\intercal} {\bf x} > b_j\}$, and compress the intersection of half-spaces. We illustrate this novel approach in Fig. \ref{fig:hyperbinning}. 

The hyper binning-based encoding scheme is applicable under broad source distributions $p_{X_1 ,\dots,X_s}(x_1,\dots,x_s)$, given a number of sources $s$. In Sect. \ref{HyperBinDesign}, we use the GMM as a tractable instance under the general framework for hyper binning. 
We consider the case in which the encoders have the same parameters $\{{\bf a}_j,\,b_j\}_{j=1}^J$ motivated by using the G{\'a}cs-K{\"o}rner CI, where the CI rate is the rate of compressing the parameters $\{{\bf a}_j,\,b_j\}_{j=1}^J$ (as will be detailed in Sect. \ref{HyperBinning}). While the hyperplane parameters for the individual encoders need not be the same, the CI between the encoders is less for the case of different parameters versus the same parameters.

\subsection{Why Hyper Binning?}\label{why_binning}
Sending colorings of sufficiently large power graphs of characteristic graphs followed by source coding, e.g., Slepian-Wolf compression \cite{slepian1973noiseless}, leads to an achievable encoding for compressing functions provided that the functions satisfy some additional conditions \cite{feizi2014network}. 
Instead of sending source variables, it is optimal to send coloring variables that model a valid encoding of a characteristic graph that captures which source outcomes should be distinguished to recover the desired function \cite{OrlRoc2001}. The destination then uses a look-up table to compute the desired function value by using the received colorings. While in some cases, the coloring problem is not NP-hard, in general, finding this coloring is an NP-complete problem  \cite{cardinal2004minimum}. Interesting instances include set cover \cite{feige1998threshold}, 3-dimensional matching \cite{karp1972reducibility}, and 3-coloring \cite{eppstein2000improved}. 

As in Slepian-Wolf encoding, in hyper binning, each bin represents a typical sequence of function $f$'s outcomes and is a collection of infinite length sequences.  
Hyper binning does not rely on NP-hard concepts such as finding the minimum entropy coloring of the characteristic graph of $f$. 
Unlike graph coloring, hyper binning with a sufficient number of hyperplanes in GP jointly partitions the source random variables in a way to achieve the desired quantization error at the destination for a given computation task. 
Given an entropy-based distortion measure as in, e.g., \cite{courtade2011multiterminal}, we require the following condition on the number of hyperplanes $J$: $R_1+R_2=\sum\nolimits_{j=1}^{J} h(q_j)\geq 1-\epsilon$
\begin{align}
J = \min_{k} \Big\{k\,: \,  \sum\limits_{j=k+1}^{2n} h(q_j)\leq\epsilon\Big\}.  \nonumber
\end{align}
Hyper binning naturally allows (a conditionally) independent encoding across the sources via an ordering of hyperplanes at each source prior to transmission and their joint decoding at the destination. This is possible with a helper mechanism that ensures the communication of the common randomness, through extracting the G{\'a}cs-K{\"o}rner common information (GK-CI) \cite{gacs1973common}, characterized via hyperplanes.   
The GK-CI variable is 
the maximum common
information that can be extracted from each source. We detail this measure in Sect. \ref{HyperBinning}.

\subsection{Technical Background}
\label{convex_background}

The next remark provides the necessary and sufficient condition for a set to be convex. It also imposes a necessary condition on the function $f$ we represent via partitioning.

\begin{remark}\label{convex_iff}
A set $C$ is convex if and only if for any random variable $X$ (or function) over $C$, $\mathbb{P}(X \in C) = 1$, its expectation is also in $C$, i.e., $\mathbb{E}[X] \in C$ \cite{dattorro2010convex}. 
\end{remark}

If $C = \{P_C\,{\bf x}\, \vert\, {\bf x} \in \mathbb{R}^n\}$, then for each ${\bf x}\in\mathbb{R}^n$ there exists a unique point $P_C\,{\bf x}\in C$ 
that is closest to ${\bf x}$ in the Euclidean sense. Unique projection of ${\bf x}$ onto $C$ \cite[Ch. E.9]{dattorro2010convex} equals
\[
P_C\,{\bf x}  = \arg\min\limits_{{\bf y}\in \mathcal{C}} \normtwo{ {\bf x} - {\bf y} }.
\]

Every linear hyperplane $\eta$ is an affine set parallel to an $(n-1)$-dimensional subspace of $\mathbb{R}^n$ \cite{dattorro2010convex}. Let $\mathcal{H}^n$ be the space of hyperplanes in $\mathbb{R}^n$. A hyperplane $\eta \in \mathcal{H}^n \subset \mathbb{R}^n$ is characterized by the linear relationship given as follows: 
\begin{align}
\label{HyperplaneEquation}
    \eta({\bf a}, b) = \{{\bf y} \in \mathbb{R}^n: {\bf a}^{\intercal}{\bf y} = b\},\,\, {\bf a} \in S^{n-1},\,\, b \in \mathbb{R},
\end{align}
where ${\bf a}$ is the nonzero normal and $S^{n-1}$ is the unit sphere.

Projection of ${\bf x}\in\mathbb{R}^n$ onto $\eta$ \cite[Ch. E.5]{dattorro2010convex} is given as
\[
P{\bf x} = \arg\min\limits_{{\bf y}\in H} \normtwo{ {\bf x} - {\bf y} }\\
= {\bf x} - {\bf a}({\bf a}^{\intercal}{\bf a})^{-1}({\bf a}^{\intercal}{\bf x}-b).
\]

We shall let $s$ and $J$ denote the number of sources and hyperplanes, respectively. 
A hyperplane arrangement of size $J$ in an $s$ dimensional source space creates at most $r(s,J)=\sum\nolimits_{k=0}^s \binom{J}{k}\leq 2^J$ regions. Hyperplanes in general position (GP) divide the space to $r(s,J)$ regions \cite{chataignon2019comparison}. In this paper, for the sake of presentation, we study the case $s=2$. We represent the source data by feature vectors $\{{\bf x}_t\}\in \mathbb{R}^n$ that are mixtures of Gaussian variables and sampled from a training set where $t$ denotes the index of ${\bf x}_t$, which we detail in Sect. \ref{HyperBinDesign}.

\begin{ex}\label{GP}
A hyperplane arrangement of size $J=3$ for $n=2$ in GP divides the space into $r(2,3)=7$ regions. 
\end{ex}

For each hyperplane $\eta({\bf a}, b)$ there are $n+1$ unknowns ${\bf a},\, b$  to be determined, hence there are $(n+1)J$ unknown hyperplane parameters in total. A given number of hyperplanes $J$ in GP can support a feature vector in an $n$-dimensional space where the dimension is upper bounded as
\begin{align}
\label{nmax}
   n_{\max}=\max\limits_{n\geq 1}\,\,[n \,\vert\, (n+1)J \leq r(s,J)]. 
\end{align}
\begin{figure}[t!]
\centering
\includegraphics[width=0.5\textwidth]{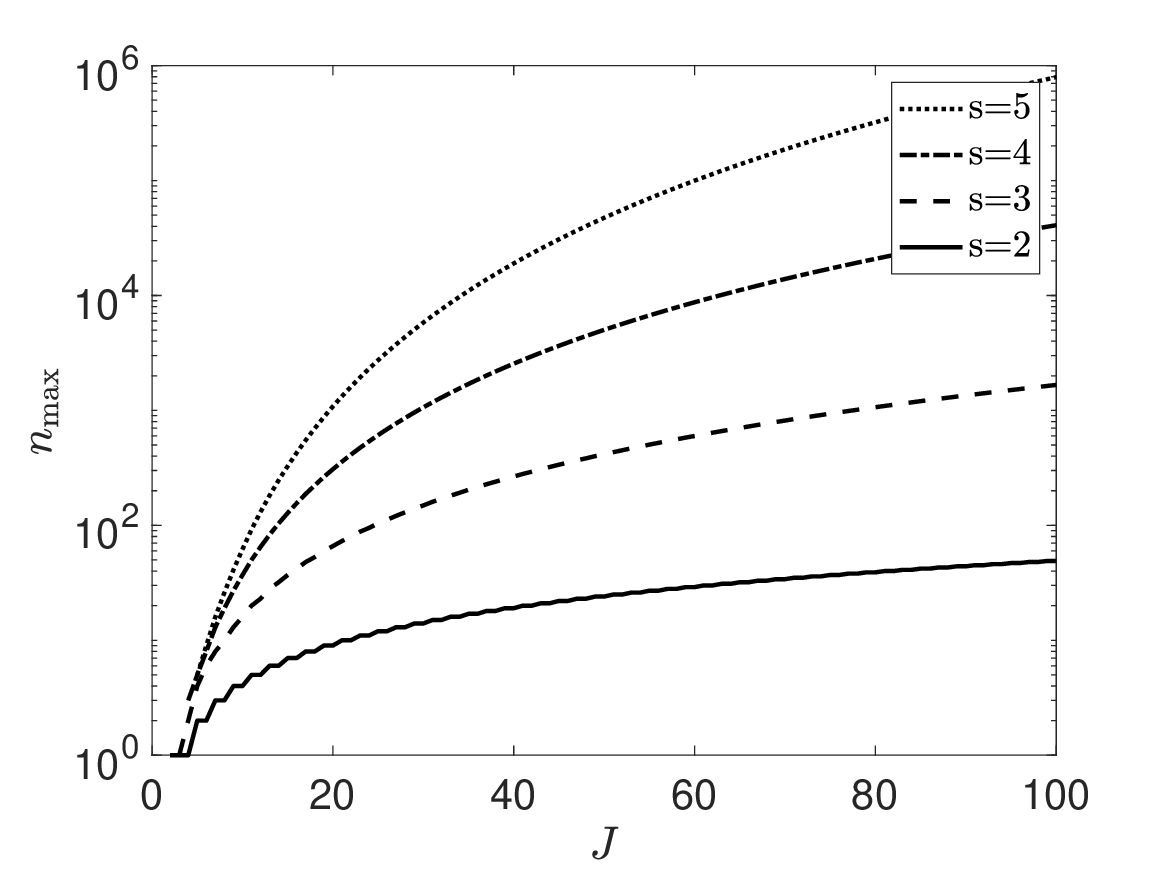}
\caption{\small{Maximum $n$ vs $J$ in GP for different number of sources $s$.}}\label{fig:nvsJ}
\end{figure}

This inequality is because the required number of hyperplane parameters, $(n+1)J$ unknowns, should be smaller than the number of regions denoting the quantized function outcomes where each output is an equation represented by intersections of half-spaces. This result gives a necessary condition (a lower bound on $J$) to support a feature vector in $\mathbb{R}^n$. In Fig. \ref{fig:nvsJ} we sketch the relation between $n_{\max}$ and $J$.  
The number of dimensions a hyperplane arrangement can capture scales exponentially in the number of sources $s$ (for $s>2$) vs orthogonal binning that provides linear scaling of the total number of dimensions in $s$ as $J$ increases. In this paper, $s=2$, and  $n_{\max}=\frac{J}{2}+O\left(\frac{1}{J}\right)$ following from (\ref{nmax}), i.e., $J$ linearly scales with $n$ as $J\approx 2n_{\max}$ as $n$ tends to infinity.  
\begin{theo}{\bf Kolmogorov-Arnold representation theorem \cite{kolmogorov1957representation}.} Given a multivariate continuous function $f$, Kolmogorov and Arnold established that $f$ can be written as a finite composition of continuous functions of a single variable and the binary operation of addition \cite{arnolddessert}, \cite{kolmogorov1957representation}. More specifically, 
\begin{align}
\label{KolmogorovArnold_representation}
f({\mathbf{x}^n})=f(x_1,\dots,x_n)=\sum\limits_{j=0}^{2n} \Psi_j\Big(\sum\limits_{p=1}^n \psi_{j,p}(x_p)\Big).  
\end{align}
\end{theo}
For the representation result of Kolmogorov-Arnold, there exist proofs with specific constructions \cite{braun2009constructive}. In (\ref{KolmogorovArnold_representation}) the only true multivariate function is the sum since every function can be written using univariate functions and summing \cite{kolmogorov1957representation}. Hyper binning is a special case of (\ref{KolmogorovArnold_representation}) such that $\psi_{j,p}(x_p)=a_{jp}\cdot x_p$ for some constant $a_{jp}$, for $j\in\{0,\dots,2n\}$ and $p\in\{1,\dots,n\}$, where we approximate the upper limit $2n$ of the outer sum with $J$ hyperplanes. We note that for a class of functions we are interested in, it is possible that a subset of outer functions $\{\Psi_j\}_{j=0}^{2n}$ may be equal to $0$. For a continuous function given in the most general form in (\ref{KolmogorovArnold_representation}), the approximation is more accurate when the hyperplanes are selected to capture the high-influence $\Psi_j$ terms, meaning the terms that have a more dominant impact on $f$. Furthermore, we model ${\mathbf{x}^n}$ using feature vectors $\{{\bf x}_t\}$ coming from either source, which we detail in Sect. \ref{HyperBinDesign}. 
The outer function $\Psi_j$ satisfies
\begin{align}
\label{outer_function}
\Psi_j(y)=c_j\cdot \mathbbm{1}_{y \geq b_j}+d_j\cdot \mathbbm{1}_{y < b_j},
\end{align}
where $c_j$ and $d_j$ are constants with opposite signs to represent the two different directions of hyperplane $j$.

Combining the upper bound in (\ref{nmax}), which states $n_{\max}\approx \frac{J}{2}$, and the upper limit $2n$ of the sum in (\ref{KolmogorovArnold_representation}), we see that choosing $J=2n$ makes sure that in the representation theorem of Kolmogorov-Arnold in \cite{kolmogorov1957representation}, the error of representation vanishes as $n\to\infty$. The linear scaling between $n$ and $J$ for $s=2$ in the asymptotic regime implies that random binning and hyper binning have similar performance for quantized variables as $n\to\infty$. These schemes differ at finite blocklengths since orthogonal binning requires quantization, whereas hyper binning provides an already quantized representation, where the bins are determined by a hyperplane tessellation capturing $f$. The half-space decision probabilities are parametrized by the hyperplanes (to be given by (\ref{pi_function})), and the decision regions of hyperplanes are described via binary entropy functions of these probabilities, eliminating the need for post-quantization. In orthogonal binning, quantization is followed by binning, where the latter corresponds to the post-quantization phase. Unlike in orthogonal binning, in hyper binning, the initial quantization phase via hyperplanes eliminates the need for further binning. Thus, there is no need for post-quantization. We will validate the choice of $J$ (i.e., the number of hyperplanes) to ensure the desired distortion in an MSE sense in Sect. \ref{distortion}, and tie the modeling of  finite blocklengths to Kolmogorov complexity \cite[Ch. 7.3]{cover2012elements} in Sect. \ref{compression_finite_blocklength}, respectively.

Hyper binning embeds the quantization phase and is discrete, i.e., it does not require further post-quantization prior to compression. Hence, we can apply the scheme of Berger-Tung, detailed in \cite{tung1978multiterminal} and \cite{berger1978multiterminal},  
or the coding scheme in \cite{feizi2014network} and its generalization in \cite{basu2020hypergraph} on the quantized representation. Since the compression gain of hyper binning over random binning lies in the pre-quantization aspect, any compression scheme, such as  \cite{tung1978multiterminal,berger1978multiterminal,feizi2014network,basu2020hypergraph}, can be implemented on top of hyper binning to recover the function representation in (\ref{KolmogorovArnold_representation}).

We next give a fundamental separation result on convex sets.

\begin{theo}\label{separation}{\bf Supporting hyperplane \cite[Thm 1]{miao2010projection}.} 
A point ${\bf x}$ lies in $C$ if and only if $\max\limits_{\bf a} ({\bf a}^{\intercal} {\bf x}-S_C({\bf a}))\leq 0$, where $S_C({\bf a})=\sup\{{\bf a}^{\intercal}{\bf z}\,:\, {\bf z}\in C\}$ is the supporting hyperplane. 
\end{theo}

We next detail the necessary conditions for distributed functional compression of sources via hyper binning.

\subsection{Necessary Conditions for Encoding the Functions} 
\label{NecessaryConditions}  
Let $X_1$ and $X_2$ be real-valued source random variables, i.e.,  $\mathcal{X}_i\subseteq \mathbb{R}$ for $i\in\{1,2\}$. 
Our approach entails some assumptions on the function. 
Hyper binning yields a partitioning of the joint sources' data to convex sets $P_k$, where $k\in\{1, \ldots, M\}$ such that $M$ is the total number of partitions obtained from a set of hyperplanes in GP. Given a number of hyperplanes in GP, the hyperplanes attain the maximum number of partitions $M$. 

This section  details the necessary conditions for the function for encoding. 
We emphasize that the function $f$ could be continuous or discrete. Our goal is to describe a class of continuous or discrete functions via their quantized estimates $\hat{f}$.   
When $f$ is continuous, the distortion between $f$ and $\hat{f}$ is bounded away from $0$. For this case, we refer the reader to Example 2, where we contrast the achievable rate reduction performed by different binning schemes, and 
to Example 5, where we consider Gaussian variables and analyze the rate-distortion function using the notion of $\epsilon$-achievable hypergraphs \cite{basu2020functional}. 
When $f$ is discrete, which is pertinent in information theory, we refer the reader to Examples 3, 4, 6, and 7 for various discrete-valued random variables and their functions using the notions of $\epsilon$-characteristic hypergraphs versus $D$-characteristic graphs, as detailed in \cite{feizi2014network,basu2020hypergraph}. The quantized function 
$\hat{f}:(\mathcal{X}_1,\mathcal{X}_2)\to \mathcal{Z}$ is such that the mapping $\mathcal{Z}\to \{1,\hdots, M\}$ is a bijection. 
From Remark \ref{convex_iff}, our model is restricted to a class of functions $f$ satisfying that if $\mathbb{P}(f \in P_k) = 1$, then $\mathbb{E}[f] \in P_k$ for each $P_k$ \cite{Iosif12}.

The function $f$ to be quantized, if continuous, has to be continuous at $(x_1,\,x_2)\in (\mathcal{X}_1,\,\mathcal{X}_2)$ since $f^{-1}(P_k)$ is a neighborhood of $(x_1,\,x_2)$ for every neighborhood $P_k$ of $f(x_1,\,x_2)$ in $\mathcal{Z}$.  
If it is not continuous at $(x_1,\,x_2)$, there might be a region $P_k$ for a given $k\in\{1, \ldots, M\}$ of $f(x_1,\,x_2)$ such that $f^{-1}(P_k)$ is not a neighborhood of $(x_1,\,x_2)$. In other words,  
multiple disjoint hyper bins may  
yield the same function outcome, which we do not explicitly capture in our setting. The domain of $f$ can also be discrete, as in \cite{feizi2014network}.

Using the basic properties of hyperplanes presented in this section, in Sect. \ref{HyperBinDesign}-(A-C) we will develop the scheme of {\em hyper binning} where the sources can be described by a Gaussian mixture model, a specific tractable instance of the general problem stated in Sect. \ref{problem}. The coding rate of hyper binning can be characterized using the binary entropy function, which we will detail.
We will also discuss the hyperplane arrangements for different rate-distortion models and demonstrate the  
performance of hyper binning versus the state-of-the-art solutions.

\section{Designing Hyper Bins for Gaussian Mixtures}
\label{HyperBinDesign}
We design hyper binning for distributed functional quantization recalling that the source data is represented by feature vectors.  
We use linear discriminant analysis (LDA) to distinguish different classes of feature vector combinations yielding the same function outcome. LDA is a classification technique for the separation of multiple classes of variables using linear combinations of observations/features/measurements, where the classes are known a priori. LDA works when the measurements on independent variables for each observation are continuous quantities.  
The set of features $\{{\bf x}_t\}$ for each sample of an event has a known class $y$. The classification problem is to find a good predictor for the class $y$ of any sample of the same distribution 
given only an observation ${{\bf x}_t}$. In LDA the conditional probability density functions (pdfs) $p({{\bf x}_t}|y=0)$ and $p({{\bf x}_t}|y=1)$ are both the normal distribution with mean and covariance parameters $\left({\bm \mu}_0,\Sigma\right)$ and $\left({\bm \mu}_{1},\Sigma\right)$, respectively. 
Hence, the samples come from a Gaussian mixture model (GMM) given by $\phi_0\,\mathcal{N}({\bm \mu}_{0},\Sigma)+\phi_1\,\mathcal{N}({\bm \mu}_{1},\Sigma)$, where $\phi_0$ and $\phi_1$ are the cluster weights given by the proportion of feature vectors. The Bayes optimal solution is to predict points as being from the second class if the log-likelihood ratio is bigger than some threshold, so that the decision criterion of ${{\bf x}_t}$ being in a class $y$ is a threshold on the dot product ${\bf w}\cdot {{\bf x}_t}>c$, i.e., a function of the linear combination of the  observations, for some threshold $c$, where ${\bf w}=\Sigma ^{-1}({\bm \mu}_{1}-{\bm \mu}_{0})$ and $c={\bf w}\cdot {\frac {1}{2}}({\bm \mu}_{1}+{\bm \mu}_{0})$.   
The observation belongs to $y$ if the corresponding ${{\bf x}_t}$ is located on a certain side of a hyperplane perpendicular to ${\bf w}$. The location of the plane is defined by $c$. 

In the case of multiple classes, the GMM is given by
\begin{align}
\label{GaussianMixtureModel}
W_M\sim\,\,\gamma_k\sum\limits_{k=1}^{M}  \mathcal{N}({\bm \mu}_k, \Sigma)\ ,\quad M\geq 2 \ ,  
\end{align}
where $M$ is the total number of classes, ${\bm \mu}_k$ is the mean vector for class $k\in\{1,\dots,M\}$, $\Sigma$ is the covariance matrix (same for all $k$), $n_k$ is 
the count of feature vectors $\{{\bf x}_t\}$ of class $k$ in the data, $N=\sum\nolimits_{k=1}^M n_k$ is the total number of 
$\{{\bf x}_t\}$, and $\gamma_k=\frac{n_k}{N}$ is the relative count of class $k$ data. The analysis for LDA with 2 classes can be extended to find a subspace that contains the class variability. The conditional pdfs $p({\bf x}_t|y=k)$ 
for  
$\{{\bf x}_t\}$ of 
class $k$ are independent Gaussian variables 
given by 
$\mathcal{N}({\bm \mu}_k,\Sigma)$   
(same for all $k$). The scatter between class variability is the sample covariance of the class means $\Sigma _{b}={\frac {1}{M}}\sum _{l=1}^{M}(\mu _{l}-\mu )(\mu _{l}-\mu )^{T}$, where $\mu$ is the mean of the class means. The class separation in a direction ${\bf w}$ is $S= {{\bf w}^{T}\Sigma _{b}{\bf w}}{({\bf w}^{T}\Sigma {\bf w})^{-1}}$.

\subsection{Data and Hyperplane Arrangement}
\label{data}
The feature vectors $\{{\bf x}_t\}$ lie in an $n$-dimensional space and are modeled by a GMM given by (\ref{GaussianMixtureModel}), where each ${\bf x}_t$ can come from (belong to) either source $X_1$ or $X_2$, and $\{{\bf x}_t\}$ are independent and belong to the same GMM.

We employ a linear hyperplane arrangement for classifying $\{{\bf x}_t\}$. 
To describe this arrangement, we need $J (n+1)$ parameters in total, where $J$ is the number of hyperplanes. To achieve the desired distortion for a given function $f$, we shall choose $J$ following (\ref{nmax}) to represent or distinguish the desired number of distinct outcomes $M$ of $f$. The orientations of hyperplanes will depend on the correlations between $X_1$ and $X_2$ as well as the correlations between $(X_1,\,X_2)$ and $f$.

\subsection{Optimizing Hyperplane Arrangement}
\label{HyperplaneDistribution}

To provide a joint characterization of sources by capturing their correlation as well as the features of the function $f$, we exploit LDA. In LDA, the encoded data is obtained by projecting the source data on a hyperplane arrangement, and by looking at which side of each hyperplane the vector lies. The criterion of a vector being in a class $y$ is purely a function of this linear combination of the known observations. The observation belongs to $y$ if the corresponding vector is located on a certain side of a hyperplane. We independently design each hyperplane. 
The hyperplane arrangement, i.e., the collection of hyperplane parameters $({\bf a}, b)$, depends on the particular function $f$ and the distribution of the vectors $\{{\bf x}_t\}$.

For a hyperplane $\eta({\bf a}, b)$ described by vector ${\bf a}\in\mathbb{R}^n$ and $b\in\mathbb{R}$ as in (\ref{HyperplaneEquation}), the projected feature vector $u_t ={\bf a}^{\intercal} {\bf x}_t$ lies on one side of $\eta({\bf a}, b)$ if $u_t \leq b$, and on the other side if $u_t \geq b$. Mapping ${\bf x}_t$ to the $u_t$ space is equivalent to computing the inner product of the feature vector and ${\bf a}$. As a result of this linear mapping, 
the distribution for the one-dimensional mapping outcome that models class $k$ is also Gaussian, with the following mean and variance, respectively:
\begin{align}
m_k = {\bm \mu}_k^{\intercal} {\bf a},   \quad \sigma^2={\bf a}^{\intercal}\Sigma {\bf a},\quad k=1,\hdots, M.  
\end{align}
In our setup, we note that the feature vectors lie in a high dimensional space and form an independent set of Gaussian random variables. Because their linear projections onto hyperplanes are also Gaussian and independent, the notions of the {\em set of hyperplanes} and the {\em feature vector classes} are exchangeable. More specifically, with a careful choice of the parameters $\{m_k\}_{k=1}^M$ and $\sigma$, we can observe that (i) projecting multiple vector classes onto a single hyperplane is equivalent to (ii) projecting a set of feature vectors $\{{\bf x}_t\}$ onto $M$ hyperplanes to generate a total number of classes $M$ where each class index $k$ can be considered as a mapping from $\{{\bf x}_t\}$ to a hyper bin index. Using this analogy, we represent/describe our model in (ii) via the multi-class interpretation in (i).

In the multi-class interpretation, the number of feature vectors of class $k$ that lie to the right of $b={\bf a}^{\intercal} {\bm \mu}'$ for a hyperplane characterized by $\eta({\bf a}, b)$ is given by $n_{k,r}=n_k p_k$ where the probability that a feature vector belongs to partition $k$, or equivalently it lies to the right of $b$, is given by
\begin{align}
\label{pi_function}
p_k=Q\Big(\frac{|b-m_k|}{\sigma}\Big),\quad k=1,\hdots,M,
\end{align}
where $m_k$ is the one-dimensional mapped mean of ${\bf x}_t$ to the $u_t$ space given that it belongs to class $k$, and $Q(z)=\frac{1}{2}\erfc\big(\frac{z}{\sqrt{2}}\big)$ is the complementary cumulative distribution function (CDF) of the standard Gaussian distribution such that $Q(z)\to 0$ as $z\to\infty$ monotonically. 
An  
observation is that as $\sigma$ increases, $p_k$ becomes higher due to (\ref{pi_function}). As $\sigma$ increases, since $p_k$'s also increase, $p_{M+1}$ increases. Furthermore, $p_k$ increases in $m_k$ given that $b\geq m_k$. We assume that $p_k$ is a fixed constant, and the function $f$ determines the distribution $\{p_k\}_{k=1}^M$.

In the multi-hyperplane interpretation, let $q_j=\mathbb{P}({\bf a}_j^{\intercal} {\bf x}_t\geq b_j)$ be the probability that a feature vector lies to the right of $b_j$ for a hyperplane $j=1,\hdots, J$ characterized by $\eta({\bf a}_j, b_j)$, i.e., the tail probability of one-dimensional Gaussian variable. Hence, the relation between $\{p_k\}_{k=1}^M$ and $\{q_j\}_{j=1}^J$ satisfies 
\begin{align}
\label{p_k_vs_q_j}
p_k=\prod_{j\in S_k } q_j \prod_{j\notin S_k} (1-q_j), \quad k=1,\dots, M,   
\end{align}
where $S_k$ is the set of the hyperplanes $j$ for which the hyper bin $k$ lies to the right of $b_j$.  
Our goal is to decide  
the class of hyperplanes with optimal 
$\{({\bf a}_j, b_j)\}_{j=1}^J$ 
such that if we assign the feature vectors at each source to one of two partitions based on whether $u_t\leq b$ (or equivalently ${\bf x}^{\intercal} {\bf a}\leq b$), then the average of the entropy of the class distribution in the partitions is minimized \cite{padmanabhan1999partitioning}. Minimizing the entropy of partitioning is equivalent to maximizing the mutual information associated with the partitioning, i.e., the difference between the entropy of function $f$ and the average of the entropy of the partitions.

Our objective is to minimize the entropy of the partitioning. To that end, we choose the following mutual information metric associated with the partitioning via hyper binning:
\begin{align}
\label{MI}
   I(M,\{p_k\}_{k=1}^{M+1},\{n_k\}_{k=1}^M)= 
   h(p_{M+1}) - \sum\limits_{k=1}^M \gamma_k h(p_k).
\end{align}
This metric captures the accuracy of classifying the function outcomes. While $I(\cdot)$ depends on $\{p_k\}_{k=1}^{M+1}$, $\{n_k\}_{k=1}^M$ such that $N=\sum\nolimits_{k=1}^M n_k$, we only emphasize its dependence on the number of classes $M$ in the partitioning process, i.e., use $I(M)$ for brevity. The higher the entropy for the classification of $M$ partitions, i.e., $\{h(p_k)\}_{k=1}^M$, the lower $I(M)$ is. 

The trend of $I(M)$ in (\ref{MI}) depends on the distributions of $\{{\bf x}_t\}$.  
To maximize $I(M)$ via hyper binning, it is intuitive that $({\bf a}, b)$ should be such that $p_k$'s are close to $0$ or $1$ to minimize $h(p_k)$'s, and $p_{M+1}=\frac{1}{N}\sum\nolimits_{k=1}^M n_{k,r}=\frac{1}{N}\sum\nolimits_{k=1}^M n_kp_k$ is close to $0.5$, i.e., there are an approximately equal number of feature vectors in the two partitions to maximize $h(p_{M+1})$.

Assume for optimal $I(M)$ that each $p_k$ is approximately $0$ or $1$. As $\sigma$ increases, $h(p_{M+1})$ decreases, and  
$h(p_k)$, $k\in\{1,\dots,M\}$, increases. For the asymmetric case where $n_k$ 
is proportional to $p_k$,  
incrementing $M$ improves 
$I(M)$ since each added hyperplane provides 
more information to distinguish the function outcomes. However, for the symmetric case where $n_k=\frac{N}{M}$,  
adding beyond a certain number of hyperplanes does 
not help. We later demonstrate this behavior in Prop. \ref{MI_symmetric}.

\subsection{Source Data Distribution Models for a GMM}
\label{source_data_distribution_symmetric_asymmetric}
We next assume that $m_1<\hdots<m_{M-1}<m_M$, and $\sigma$ is fixed. Hence, (\ref{pi_function}) implies that $p_1<\hdots<p_{M-1}<p_M$. Further assume that $\frac{1}{2}<p_k$ for all $k$, yielding $h(p_1)>\hdots>h(p_M)$.

\subsubsection{Asymmetric Data Distribution} 
\label{asymmetric}
Consider a scenario where the class distribution is such that each $n_k$ is proportional to $p_k$, i.e., $n_k = \beta p_k$ for some $\beta\in\mathbb{R}^{+}$. The asymmetry among $p_k$ is exacerbated by $n_k$. This asymmetry makes the classes more distinguishable. A function that satisfies this criterion, i.e., a non-surjective function, can be compressed well. 

We next determine a relation between the mutual information metrics $I(M)=I(M,\{p_k\}_{k=1}^{M+1},\{n_k\}_{k=1}^M)$ versus $I(M+1)=I(M+1,\{p_k\}_{k=1}^{M+2},\{\tilde{n}_k\}_{k=1}^{M+1})$ when we switch from $M$ to $M+1$ classes. To that end, we assume that (i) $n_k = \beta p_k$, $k\in \{1,\dots,M\}$ for $\beta\in\mathbb{R}^{+}$, (ii) the number of feature vectors of class $k$ in the source data distribution $W_{M+1}$, namely $\{\tilde{n}_k\}_{k=1}^{M+1}$ versus $\{n_k\}_{k=1}^M$, satisfies $\tilde{n}_k=\alpha n_k$ for some $\alpha\in[0,1]$, and (iii) the  
distribution $W_{M+1}$ is such that 
$\{p_k\}_{k=1}^{M+2}$ can be determined using $p_{M+2}=\frac{1}{N}\sum\nolimits_{k=1}^{M+1} \tilde{n}_k p_k$, noting that inclusion of the $M+1$-th class does not change the probability $\{p_k\}_{k=1}^M$ given in (\ref{pi_function}) that a feature vector (of class $k$) lies to the right of a hyperplane characterized by $\eta({\bf a}, b)$.

\begin{prop}\label{MI_lowerbound}
Consider a setting with two distributed sources and a function $f$ satisfying the  properties listed in Sect. \ref{NecessaryConditions}. Then for two data distributions $W_{M}$ and $W_{M+1}$ of the form (\ref{GaussianMixtureModel}) that are related via (i)-(iii) outlined above, we have the relation $I(M+1)\geq I(M),\,\, \forall M\geq 1$.
\end{prop}

\begin{proof}
We refer the reader to the supplementary material.
\end{proof}

\subsubsection{Symmetric Data Distribution}
\label{symmetric}
We now consider the uniform class distribution case such that $n_k=\frac{N}{M}$. Unlike the asymmetric case the classes are less distinguishable. A function that satisfies this criterion cannot be compressed well due to its surjectivity. For the symmetric case, let $\bar{p}_M=\frac{1}{M}\sum\nolimits_{k=1}^M p_k$. Hence, we obtain $(M+1)\bar{p}_{M+1}=M\bar{p}_M+p_{M+1}$. Letting $\bar{h}_M=\frac{1}{M}\sum\nolimits_{k=1}^{M}  h(p_k)$, it is easy to note that $(M+1)\bar{h}_{M+1}=M\bar{h}_M+h(p_{M+1})$.

\begin{prop}
\label{Bounding_MI_symmetric}
For symmetric data distributions $W_{M}$ and $W_{M+1}$ of the form (\ref{GaussianMixtureModel}) and related via (i)-(iii), it holds that 
\begin{align}
\label{MI_bounds_difference}
\hspace{-0.28cm}\frac{\bar{h}_M-h\left(\bar{p}_{M}\right)}{M+1}\leq I(M+1)-I(M) \leq \frac{\bar{h}_M-h(p_{M+1})}{M+1},
\end{align}
where we note that $h\left(\bar{p}_{M}\right)>h(p_{M+1})$. In the limit as $M\to \infty$, the gap $ I(M+1)-I(M)\to 0$ from the squeeze theorem.
\end{prop}
\begin{proof}
We refer the reader to the supplementary material.
\end{proof}

Prop. \ref{Bounding_MI_symmetric} provides a convergence result on $I(M)$ as in (\ref{MI_bounds_difference}). It also implies that the gains caused by the increments in $M$ provide diminishing returns, consistent with the intuition. 

For the uniform scenario, $I(M) =h\left(\bar{p}_{M}\right)-\bar{h}_M$. Because entropy is concave, $h(\bar{p}_M)\geq \bar{h}_M$. Let $h(\bar{p}^*_M)=\bar{h}_M$ such that $\bar{p}^*_M\leq 1/2$. This implies that $\bar{p}_M\in [\bar{p}^*_M, 1-\bar{p}^*_M]$.

\begin{prop}
\label{MI_symmetric}
For the symmetric data distribution model  with a source data distribution $W_{M}$ as in (\ref{GaussianMixtureModel}),  
$I(M)$ converges to
\begin{align}
\lim\limits_{N\to\infty} I(N) = I(1)+\sum\limits_{M=1}^{\infty}\frac{\bar{h}_M-h(p_{M+1})}{M+1}.
\end{align}
\end{prop}

\begin{proof}
The proof follows from convergence of $\{\frac{\bar{h}_M-h(p_{M+1})}{M+1}\}_M$ to $0$ as $M\to \infty$. 
For details, we refer the reader to the supplementary material.
\end{proof}

In Fig. \ref{fig:MI} (c)-(d), we illustrate the variation of $I(M)$ with respect to $M$ for different $\sigma^2$. On the left, the class distribution is asymmetric such that $n_k \propto p_k$. Here, the monotone increasing trend of $I(M)$ can be observed. On the right, the class distribution is uniform such that $n_k=N/M$. Note that $I(M)$ drops with $\sigma$ because the higher variability, the harder it becomes to distinguish the classes. We observe this trend in both cases. In the asymmetric case with $n_k\propto p_k$, we have the relation $h\left(\frac{1}{N}\sum\nolimits_{k=1}^M n_k p_k\right)>h\left(\frac{1}{M}\sum\nolimits_{k=1}^M p_k\right)$. Furthermore, $\sum\nolimits_{k=1}^M \gamma_k h(p_k)<\frac{1}{M}\sum\nolimits_{k=1}^M  h(p_k)$. Hence, $I(M)$ is always higher for the asymmetric case than for the symmetric case.

\begin{figure*}[h!]
\centering
\includegraphics[width=\textwidth]{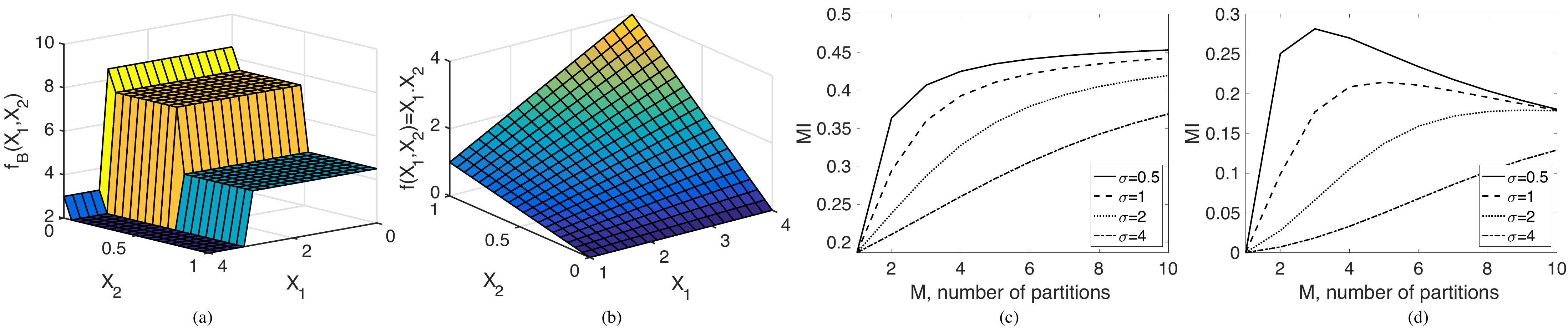}
\caption{\small{(a) A block function (no correlation across bins). (b) A smooth function (correlation across bins). 
Mutual information $I(M)$ versus $M$: (c) Asymmetric, $n_k\propto p_k$, (d) Symmetric, $n_k=N/M$.}}\label{fig:MI}
\end{figure*}

\subsection{Rate-Distortion Models for Hyper Binning}
\label{distortion}
The rate-distortion function is the solution of the problem $R(D)=\min\nolimits_{p_{\hat{X}\vert X}(\hat{x}\vert x)} \big\{I(X;\hat{X})\, :\, \mathbb{E}[d(X,\hat{X})]\leq D\big\}$, where $p_{\hat{X}\vert X}(\hat{x}\vert x)$ is the conditional probability density function (PDF) of the compressed signal $\hat{X}$ for the original signal $X$.

For hyper binning, given a distortion level $D>0$, the rate-distortion function for $f({\mathbf{X}_1^n},\,{\mathbf{X}_2^n})$ or shortly for $\mathbf{f}^n$ satisfies 
\begin{align}
\label{Rate_distortion_df}
R(D)=\min\limits_{J} \Big\{\sum_{j=1}^J h(q_j) \,:\, \mathbb{E}[d({\mathbf{f}^n},{\mathbf{\hat{f}}^n})]\leq D\Big\}.
\end{align}
Using LDA classification error we can compute $\mathbb{E}[d({\mathbf{f}^n},{\mathbf{\hat{f}}^n})]$ in (\ref{Rate_distortion_df}), using the CDF of $\{{\bf x}_t\}$ and relation (\ref{p_k_vs_q_j}). When there are two classes to be distinguished, the Bhattacharyya bound upper bounds the error probability \cite{Mazumdar2017}, \cite{duda2006pattern}. For the multi-class model, we leave the error analysis as future work.

Exploiting the representation in (\ref{KolmogorovArnold_representation}) where the outer function satisfies (\ref{outer_function}), we next consider different distortion criteria.

\paragraph{Entropy-based distortion}
\label{entropy_based_distortion}
In random binning, the entropy for the $J$-bit quantization of ${\mathbf{x}_1^n}$ is $h({\mathbf{x}_1^n})+J$, where $h({\mathbf{x}_1^n})$ denotes the differential entropy of ${\mathbf{x}_1^n}$ and $\Delta=2^{-J}$ is the bin length. For a Gaussian vector ${\mathbf{x}_1^n}$  with a covariance matrix $\Sigma$, the entropy of its $J$-bit quantization is approximately 
\begin{align}
\label{entropy_quantized_x1n}
h({\mathbf{x}_1^n}) + J={\frac {1}{2}}\log((2\pi e)^{n}\det {\Sigma})+J.    
\end{align}
On the other hand, in hyper binning, we derive a vector quantized functional representation of the data vector using $J$ hyperplanes in total. The rate needed for this procedure is
\begin{align}
\label{entropy_hyperplanes}
\sum\limits_{j=1}^J h(q_j)=\sum\limits_{j=1}^J h\Big(Q\Big(\frac{b_j-{\bf a}_j^{\intercal}{\bm \mu}}{\sqrt{{\bf a}_j^{\intercal}\Sigma {\bf a}_j}}\Big)\Big),
\end{align}
where $q_j=\mathbb{P}({\bf a}_j^{\intercal} {\mathbf{x}_1^n}\geq b_j)$ for $j=1,\dots,J$, and (\ref{entropy_hyperplanes}) is upper bounded by $J$ because $h(q_j)\leq 1$ for all $j$. Similarly, the sum rate required for the exact description of (\ref{KolmogorovArnold_representation}) with the outer function in (\ref{outer_function}) is $\sum\nolimits_{j=1}^{2n} h(q_j)=\sum\nolimits_{j=1}^{2n} h\Big(Q\Big(\frac{b_j-{\bf a}_j^{\intercal}{\bm \mu}}{\sqrt{{\bf a}_j^{\intercal}\Sigma {\bf a}_j}}\Big)\Big)$.

We next give a necessary condition on $J$ to meet the entropy-based distortion measure, e.g., similar to \cite{courtade2011multiterminal}. We note that this result is valid for distributions beyond Gaussians.
\begin{prop}\label{entropy_based_distortion_necessary_condition}
(Entropy-based distortion for continuous random variables at infinite blocklengths.) Fix an $\epsilon>0$. Given an entropy-based distortion criterion $\mathbb{E}[d({\mathbf{f}^n},{\mathbf{\hat{f}}^n})]=[h(X)-R(D)]^+\leq \epsilon$ as $n\to\infty$, using the rate needed in (\ref{entropy_hyperplanes}) for recovering the hyper binning representation of $f$ and the rate-distortion function  
in \cite{cover2012elements} for  
squared-error distortion, the number of hyperplanes $J$ is required to satisfy the condition: 
\begin{align}
J = \min_{k} \Big\{k\,: \,  \sum\limits_{j=k+1}^{2n} h(q_j)\leq\epsilon\Big\}. 
\end{align}
\end{prop}

If source $X$ is Gaussian distributed with variance $\sigma^{2}$ and memoryless, then $h(X)=\frac{1}{2}\log(2\pi e\sigma^2)$. The rate–distortion function with squared-error distortion is given by \cite{cover2012elements}:
\begin{align}
R(D)=[h(X)-h(D)]^{+}=\Big[\frac{1}{2}\log_{2}\Big(\frac{\sigma^{2}}{\sigma^2_D}\Big)\Big]^{+},
\end{align} 
where $h(D)=\frac{1}{2}\log(2\pi e \sigma^2_D)$ is the differential entropy of a Gaussian random variable $D$ with variance $\sigma_D^2$.

To capture the effect of distortion on random binning for the quantized vector ${\mathbf{x}_1^n}$ in (\ref{entropy_quantized_x1n}) where $h(X)$ is replaced by the rate-distortion function $R({\bf D})=[h({\mathbf{x}_1^n})-h({\bf D})]^{+}$ where 
\begin{align}
h({\bf D})={\frac {1}{2}}\log((2\pi e)^{n}\det {\Sigma_D})\leq \epsilon
\end{align}
for the Gaussian vector ${\bf D}$ with covariance matrix $\Sigma_D$. 
In the asymptotic regime, the number of typical codewords is approximately  $2^{\sum\nolimits_{j=1}^J h(q_j)}$ for hyper binning, versus $2^{n H(X_i)}$ typical sequences for random binning \cite{cover1975proof}, or $2^{n H_{G_{X_i}}(X_i)}$ for characteristic graph coloring \cite{feizi2014network} of source $i\in\{1,2\}$. For finite blocklengths, we will exploit Kolmogorov complexity for the quantized vector ${\mathbf{x}_1^n}$, which is to be detailed in (\ref{Kolmogorov_complexity}).

%
\paragraph{Mean-squared error (MSE) distortion}
\label{MSE_distortion}
Given an MSE distortion criterion $\mathbb{E}[d({\mathbf{f}^n},{\mathbf{\hat{f}}^n})]=\frac{1}{n}\sum\nolimits_{l=1}^n (f(l) - \hat{f}(l))^2\leq\epsilon$, the approximation using $J$ hyperplanes yields an MSE:
\begin{align}
\label{MSE_hyperplanes}
\mathbb{E}\Big[\Big(\sum\limits_{j=J}^{2n}{{c_j}_{\{{\bf a}_j^{\intercal} {\bf x}_t \geq b_j\}}-{d_j}_{\{{\bf a}_j^{\intercal} {\bf x}_t<b_j\}}}\Big)^2\Big]\leq\epsilon.
\end{align}
We next give a sufficient condition to meet the MSE criterion. 
\begin{prop}\label{MSE_sufficient_condition}
(MMSE distortion for Gaussian random variables at infinite blocklengths.) Fix an $\epsilon>0$.  
The following condition on $J$ is sufficient to meet the MSE distortion criterion $\mathbb{E}[d({\mathbf{f}^n},{\mathbf{\hat{f}}^n})]=\frac{1}{n}\sum\nolimits_{l=1}^n (f(l) - \hat{f}(l))^2\leq\epsilon$ as $n\to\infty$, provided that $d_j=-c_j$ in the MSE expression of (\ref{MSE_hyperplanes}):
\begin{align}
\sum\limits_{k=J}^{2n} c_k \leq\sqrt{\epsilon}.    
\end{align}
\end{prop}
\begin{proof}
We refer the reader to the supplementary material.
\end{proof}

We can generalize (\ref{Rate_distortion_df}) and Prop. \ref{MSE_sufficient_condition} to finite blocklengths via the notions of dispersion \cite{kostina2012fixed} that we briefly discuss next. 
\paragraph{Hamming distortion}
\label{hamming_distortion}
For equiprobable source, the symbol error rate-distortion, i.e., $d({\mathbf{x}^n},{\mathbf{\hat{x}}^n}) = \sum\nolimits_{l=1}^n 1_{\{x_l \neq \hat{x}_l\}}$, results in $\mathbb{E}[d(X_l,\hat{X}_l)]=\mathbb{P}(X_l\neq \hat{X}_l)$. In this case, the rate-dispersion function is zero, and the finite blocklength coding rate is approximated by $R(D) +\frac{1}{2}\frac{\log n}{n}+O\big(\frac{1}{n}\big)$ \cite{kostina2012fixed}.

\paragraph{Gaussian approximation}
For a stationary and memoryless source, with bounded and separable distortion, 
i.e., $d({\mathbf{x}^n}, {\mathbf{\hat{x}}^n}) =\frac{1}{n}\sum\nolimits_{l=1}^n d(x_l,\hat{x}_l)$), the coding rate can be modeled as a function of the rate dispersion $V(D)$ \cite[Thm 12]{kostina2012fixed}:
\begin{align}
\label{coding_rate_gaussian_approximation}
R(D) + \sqrt{\frac{V(D)}{n}}Q^{-1}(\epsilon) + \theta\Big(\frac{\log n}{n}\Big),  
\end{align}
where $\theta\big(\frac{\log n}{n}\big)$ in (\ref{coding_rate_gaussian_approximation}) grows asymptotically as fast as $\frac{\log n}{n}$. Here, $\theta$ is given by Eqns. (84)-(85) in \cite[Thm 12]{kostina2012fixed}, which is a more precise definition of the Big Theta $\Theta$ notation.

While $V(D)$ provides an approximation for the coding rate, in Sect. \ref{compression_finite_blocklength}, we establish a connection between Kolmogorov complexity \cite{cover2012elements} to bound the coding rate for hyper binning.

\section{ Binning for Distributed Source Coding}\label{binning_distributed_source_coding}
In this part, we detail a fundamental limit for the asymptotic compression of distributed sources followed by an achievable random binning. This type of random binning is equivalent to orthogonal quantization of typical source sequences, as we will describe in Prop. \ref{CoverBin}. We will then contrast the hyper binning scheme with other baselines that rely on random binning.

If the encoders and the decoder do not make use of the correlation between the sources, the lowest rate one can achieve for lossless compression is $H(X_i)$ for $X_i$ for $i\in\{1,2\}$. 

{\bf Slepian-Wolf Compression.} 
This scheme is the distributed lossless compression setting with source variables $X_1$ and $X_2$ jointly distributed according to $p_{X_1,\,X_2}$, where the function $f(X_1,\,X_2)$ is the identity function. In this case, the Slepian-Wolf theorem gives a theoretical bound for the lossless coding rate for distributed coding of the two statistically dependent i.i.d. finite alphabet source sequences $X_1$ and $X_2$ as \cite{slepian1973noiseless}:
\begin{align}
\label{rateregionSW}
R_{X_1} \geq H(X_1|X_2),\quad
&R_{X_2} \geq H(X_2|X_1),\nonumber\\
R_{X_1}+R_{X_2} &\geq H(X_1,\,X_2),
\end{align}
which implies that $X_1$ can be asymptotically compressed up to the rate $H(X_1|X_2)$ \cite{slepian1973noiseless}. This theorem states that making use of the correlation allows a much better compression rate to jointly recover $(X_1,\,X_2)$ at a receiver at the expense of vanishing error probability for long sequences, it is both necessary and sufficient to separately encode $(X_1, X_2)$ at rates satisfying (\ref{rateregionSW}). The codebook design is done in a distributed way, i.e., no communication is necessary between the encoders.

{\bf Random Binning.} 
Distributed codebook design for computing functions $f$ on the data $(X_1,\,X_2)$ at the receiver sites is challenging, irrespective of whether or not $X_1$ and $X_2$ are correlated. A random code construction for source compression that achieves this fundamental limit, i.e., the Slepian-Wolf rate region for distributed sources given in \cite{slepian1973noiseless}, has been provided by Cover in \cite{cover1975proof}, which we detail next. 
\begin{prop}\label{CoverBin}{\bf Cover's random binning \cite{cover1975proof}.} 
Binning asymptotically achieves zero error for the identity function $f(X_1,\,X_2)=(X_1,\,X_2)$ when the encoders assign sufficiently large codeword lengths $nR_1$ and $nR_2$ in bits to each source sequence where $R_1 > H(X_1)$ and $R_2 > H(X_2\vert X_1)$. 
\end{prop}
\begin{proof}
Here, we list the steps of random binning, detailed in \cite{cover1975proof}, for the lossless source coding for single source case:
\begin{enumerate}[leftmargin=*]
\item Each $\mathbf{x}^n\in\mathcal{X}^n$ is randomly and independently assigned an index $m(\mathbf{x}^n)\in[1:2^{nR}]$ uniformly over $[1:2^{nR}]$. Bin $\mathcal{B}(m)$ is a subset of sequences with the same index $m$. Both the encoder and decoder know the bin assignments.
\item The encoder, upon observing $\mathbf{x}^n \in \mathcal{B}(m)$, sends index $m$.
\item The decoder, upon receiving $m$, declares that $\mathbf{\hat{x}}^n$ to be the estimate of the source sequence if it is the unique typical sequence\footnote{For a typical set $\mathcal{T}_{\epsilon }^n\subset\mathcal{X}^n$, the probability of a sequence from $X^n$ being drawn from $\mathcal{T}_{\epsilon }^n$ is greater than $1-\epsilon$, i.e., $\mathbb{P}[\mathbf{x}^n\in \mathcal{T}_{\epsilon }^n]\geq 1-\epsilon$ \cite[Ch. 3]{cover2012elements}.} in $\mathcal{B}(m)$; otherwise, it declares an error.
\item A decoding error occurs if $\mathbf{x}^n$ is not typical, i.e., $\mathcal{E}_1=\{\mathbf{X}^n \notin \mathcal{T}_{\epsilon}^n\}$, or there are multiple typical sequences, 
i.e., $\mathcal{E}_2=\{\mathbf{\tilde{x}^n}\in\mathcal{B}(M)\,\,\rm{for\,\, some}\,\, \mathbf{\tilde{x}}^n\neq \mathbf{X}^n,\,\,\mathbf{\tilde{x}}^n\in \mathcal{T}_{\epsilon}^n\}$.
\item Let $M\sim {\rm Unif}[1 : 2^{nR}] \independent \mathbf{X}^n$ denote the random bin index of $\mathbf{X}^n \in \mathcal{B}(M)$. If $R > H(X) +  \delta(\epsilon)$, Cover has shown that the probability of error $P_e^n$ averaged over $\mathbf{X}^n$ and random binnings $\to 0$ as $n\to \infty$ \cite{cover1975proof}. Hence, there is at least a sequence of binnings with $P_e^n\to 0$ as $n\to \infty$.
\end{enumerate}
The result can easily be generalized to distributed sources.
\end{proof}

To illustrate the gains that we can achieve with an optimally designed hyper binning scheme and contrast with the existing well-known binning methods, we next devise an example. Our goal is to explore how informative different types of partitionings can be for quantifying a function.

\begin{ex}{\bf Contrasting different binning methods for distributed source coding for functional compression.}
\label{Different_binning_examples}
Consider a functional compression problem where the sources $X_1$ and $X_2$ are continuous-valued. We consider three ways of compressing the sources to recover an approximate representation at the decoder. 
While random binning is asymptotically optimal, for ease of exposition, we first assume that the blocklength satisfies $n=1$. To indicate their main features, we illustrate the encoding for different binning schemes in Fig. \ref{fig:HP}, where $X_1\in[0,1]$ and $X_2\in[0,1]$ that are both uniformly distributed, and that lie on the $y$ and $x$-axes, respectively. For example, in Slepian-Wolf encoding (Left), each source independently and uniformly partitions the source outcome into $4$ bins. Hence, there are $4\times 4=16$ bins in total.
The block binning scheme (Middle) trims some of the bins in the encoding scheme of Slepian-Wolf because the function is piecewise constant or block, and there is no correlation across bins. This approach modularizes the encoding into uniform quantization and compression (bin trimming). In this example, there are $4$ blocks and each $B_k$ can be obtained via aggregating the bins of Slepian-Wolf.
If the function is more general than a block function, orthogonal trimming may not work. Instead, hyper binning can leverage the function and its dependency on the jointly distributed sources via the regions created from the intersections of linear hyperplanes and can make the quantization phase function-oriented, where the hyperplane parameters $\{({\bf a}_j,\,b_j)\}_{j=1}^J$ are adjusted according to the function $f(X_1,X_2)$. As a result, this reduces the redundancy in compression because the quantization is tailored for recovering the intended function and is more effective. 
We next detail each binning scheme separately. We emphasize that for illustration purposes, we chose $n=1$.

({\bf Top}) {\bf Binning approach of Slepian-Wolf \cite{slepian1973noiseless}.} In the first scenario, the sources first uniformly (scalar) quantize ${\mathbf{x}_1^n}\in[0,1]^n$ and ${\mathbf{x}_2^n}\in[0,1]^n$ into a discrete set using $2$ bits each. The bin assignments $(m_1({\mathbf{x}_1^n}),\,m_1({\mathbf{x}_2^n}))\in [1:4]\times [1:4]$ for the source pair $({\mathbf{X}_1^n},\,{\mathbf{X}_2^n})$ takes $M=16$ possible outcomes, with each outcome being equally likely. 
The Slepian-Wolf encoding scheme distinguishes all possible jointly typical outcomes. However, the binning scheme does not capture the function's structure, i.e., it does not distinguish $f({\mathbf{X}_1^n},\,{\mathbf{X}_2^n})$ and $({\mathbf{X}_1^n},\,{\mathbf{X}_2^n})$ from each other. In this case with $M=16$ equally likely partitions (bins), $\mathbb{P}(({\mathbf{X}_1^n},\,{\mathbf{X}_2^n})=(i_1,\,i_2))=1/16$, the entropy of the partitions equals $H(\mathbf{X}_1^n,\,{\mathbf {X}_2^n})=\log_2(16)=4$. Then, $I_{SW} = H({\mathbf{X}_1^n},\,{\mathbf{X}_2^n})-H({\mathbf{X}_1^n},\,{\mathbf{X}_2^n})=0$. We show the block diagram for independent encoding and joint decoding of two correlated data streams ${\mathbf{X}_1^n}$ and ${\mathbf{X}_2^n}$ in Fig. \ref{fig:SW}.

({\bf Left}) {\bf Orthogonal trimming of the binning-based codebook.} When the function (on $[0,1]^2$) is piecewise constant in the blocks domain, then the uniform (scalar) quantization followed by trimming achieves an optimal encoding rate. The block binning or generalized orthogonal binning scheme can capture functions with the pair $({\mathbf{X}_1^n},\,{\mathbf{X}_2^n})$ having a blockwise dependence, such as the function shown in Fig. \ref{fig:MI} (a). In this example, there are $4$ blocks $B_k$, with indices $k=1,\hdots, 4$, corresponding to different function outcomes. Hence, $f_B({\mathbf{X}_1^n},\,{\mathbf{X}_2^n})$ and $({\mathbf{X}_1^n},\,{\mathbf{X}_2^n})$ can be distinguished under this blockwise partitioning. This encoding scheme is easy to implement by combining some of the blocks prior to implementing the Slepian-Wolf encoding scheme in each $B_k$. Clearly, this is more efficient than completely ignoring the function's structure and directly implementing the Slepian-Wolf encoding. Hence, for  sources sharing blockwise dependency, i.e., $H(f_B({\mathbf{X}_1^n},\,{\mathbf{X}_2^n}))<H({\mathbf{X}_1^n},\,{\mathbf{X}_2^n})$. In this example with $4$ blocks, we use $3$ hyperplanes, as shown in Fig. \ref{fig:HP} (Middle). Hence, for block binning $\mathbb{P}(B_k)=\mathbb{P}(f_B({\mathbf{X}_1^n},\,{\mathbf{X}_2^n})=k)=\sum\nolimits_{i_1,i_2 :\, f_B=k} p_{i_1\, i_2}$. The colored region $B_2$ has a probability $\mathbb{P}(B_2)=9/16$. Similarly, $\mathbb{P}(B_1)=3/16$, $\mathbb{P}(B_3)=\mathbb{P}(B_4)=2/16$. This implies that the entropy of the partitions equals $H(f_B({\mathbf{X}_1^n},\,{\mathbf{X}_2^n}))=1.67$. In this case, block binning yields $I_B=H({\mathbf{X}_1^n},\,{\mathbf{X}_2^n})-H(f_B({\mathbf{X}_1^n},\,{\mathbf{X}_2^n}))=2.33$. We show the block diagram for orthogonal trimming-based compression for piecewise constant functions $f_B({\mathbf{X}_1^n},{\mathbf{X}_2^n})$ in Fig. \ref{fig:orthogonal_binning}. 

\begin{figure}[h!]
\centering
\includegraphics[width=0.65\textwidth]{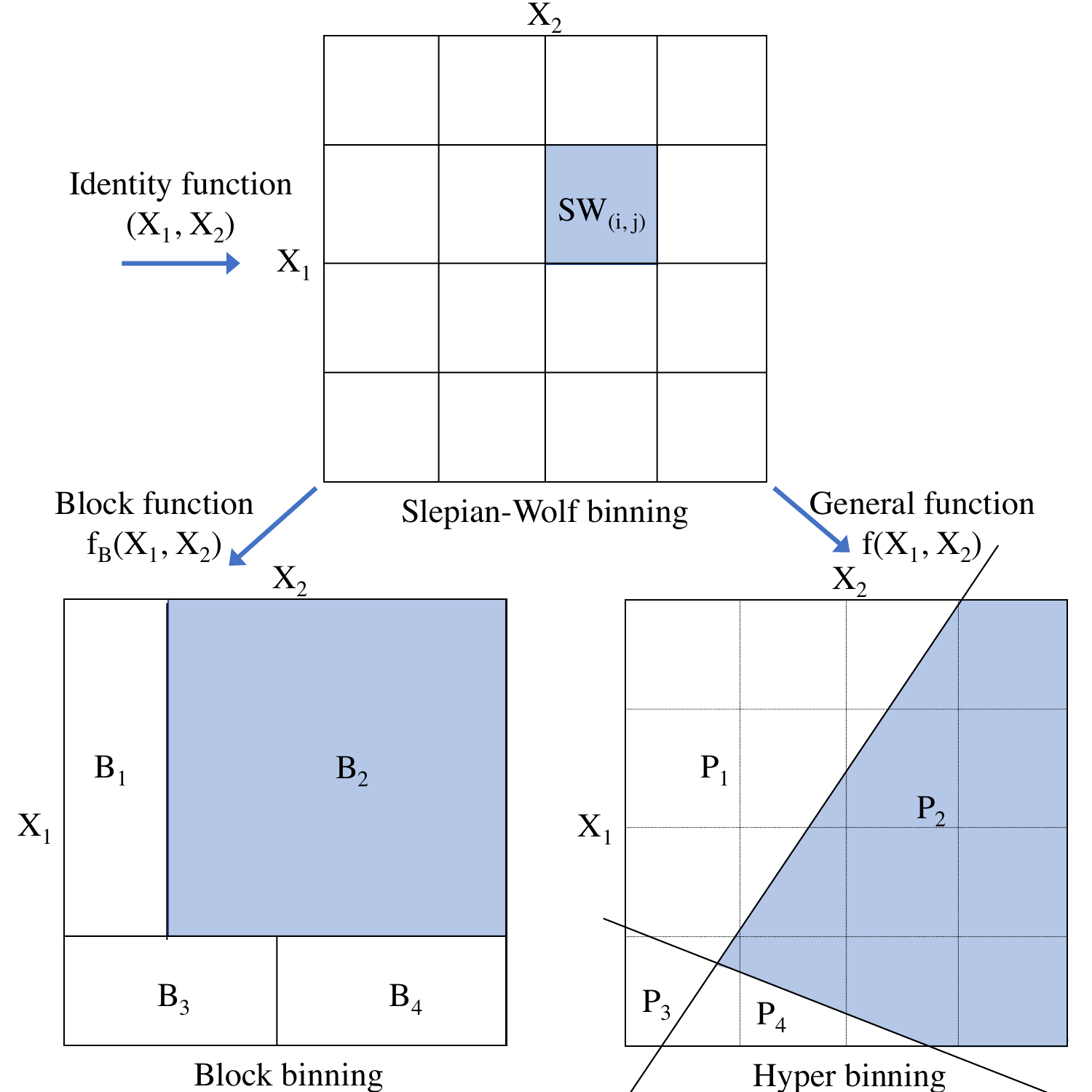}
\caption{\small{Hyperplane organization. (Top) Binning approach of Slepian-Wolf \cite{slepian1973noiseless}. 
(Left) Function sensitive, correlation insensitive partitioning. (Right) Function and correlation sensitive partitioning.}}\label{fig:HP}
\end{figure}

({\bf Right}) {\bf Hyper binning-based codebook.} If the function is not piecewise constant, then quantizing and then compressing may not be as good. The hyper binning scheme can capture the dependencies in the pair $({\mathbf{X}_1^n},\,{\mathbf{X}_2^n})$ and $f({\mathbf{X}_1^n},\,{\mathbf{X}_2^n})$, unlike the block binning scheme. In this scheme, we cannot consider the partitions $P_k$, with indices $k=1,\hdots, 4$, corresponding to function outcomes independently since each partition shares a non-orthogonal boundary to capture the dependency across the sources. With hyper binning, it is possible to jointly encode correlated sources as well as the function up to some distortion, determined by the hyperplane arrangement. As a result, for sources with dependency (more general than blockwise dependency), we can achieve $H(f({\mathbf{X}_1^n},\,{\mathbf{X}_2^n}))<H(f_B({\mathbf{X}_1^n},\,{\mathbf{X}_2^n}))$. We partition the region using $2$ hyperplanes in GP by incorporating the correlation structure between the function and the sources. In this case, $\mathbb{P}(P_1)=0.375$, $\mathbb{P}(P_2)=0.531$, $\mathbb{P}(P_3)=0.031$, $\mathbb{P}(P_4)=0.063$, and the entropy of the partitions satisfies $H(f({\mathbf{X}_1^n},\,{\mathbf{X}_2^n}))=1.42$ for each $k$. Hence, the hyper binning model yields $I(M)=H({\mathbf{X}_1^n},\,{\mathbf{X}_2^n})-H(f({\mathbf{X}_1^n},\,{\mathbf{X}_2^n}))=2.58$.  
For the example function with unit blocklength, i.e., $n=1$, as shown in Fig. \ref{fig:HP} (right), the x-axis intercepts are $0.67$ and $0.09$, and y-axis intercepts are $0.27$ and $-0.13$, and  $f:[0,1]^2\to \{1,2,3,4\}$. More specifically, $P_k$, $k=1,2,3,4$ specifies $f(x_1,x_2)$: 
\begin{align}
\label{hyperplane_example}
f(x_1,x_2)=k,\quad c_1 x_1+c_2 x_2,\,\, c_3 x_1+c_4 x_2\in P_k,
\end{align}
which is equivalent to $f(x_1,x_2)=1 \iff c_1 x_1+c_2 x_2>d_1,\,\, c_3 x_1+c_4 x_2>d_2$, and similarly for $f(x_1,x_2)\in\{2,3,4\}$, 
where $c_1=\frac{1}{0.27}$, $c_2=\frac{1}{0.67}$, $d_1=1$, and $c_3=0.68$, $c_4=-1$, $d_2=-0.09$, where the hyperplane parameters are such that the outcomes are as shown in Fig. \ref{fig:HP} (Right).
By letting $a_2=\frac{c_2}{c_1}-\frac{c_4}{c_3}$ and $b_2=\frac{d_1}{c_1}-\frac{d_2}{c_3}$, and $a_1=\frac{c_1}{c_2}-\frac{c_3}{c_4}$ and $b_1=\frac{d_1}{c_2}-\frac{d_2}{c_4}$, we can rewrite the RHS of (\ref{hyperplane_example}) for $k=1$ as
\begin{align}
\label{hyperplane_example_numbered}
f(x_1,x_2)=   
1 \iff a_1 x_1 > b_1,\,\, a_2 x_2 > b_2,
\end{align}
and similarly for $k\in\{2,3,4\}$, showing that we can reliably compute $f$ by using one hyperplane per source, i.e., $a_i x_i = b_i$, $i\in\{1,2\}$, even for $n=1$. For this example, we cannot characterize $f$ using block binning as illustrated in Fig. \ref{fig:HP} (Middle). That is because each function outcome is jointly decided. More specifically, given an outcome $f\in\mathcal{S}$, for random binning, we cannot find a disjoint set pair $\mathcal{S}_1$ and $\mathcal{S}_2$ such that $\mathbb{P}(f(X_1,X_2)\in\mathcal{S})\approx \sum\nolimits_{m_1(x_1)\in \mathcal{S}_1}\sum\nolimits_{m_2(x_2)\in\mathcal{S}_2}p(x_1,x_2)$. Hence, for $f(x_1,x_2)$ in (\ref{hyperplane_example}), hyper binning has higher accuracy than orthogonal binning in finite blocklengths $n$.  
While we can generalize hyper binning to $n\geq 2$, we next focus on the complexity of finite blocklengths due to space constraints.

\begin{figure*}[t!]
\centering
\includegraphics[width=\textwidth]{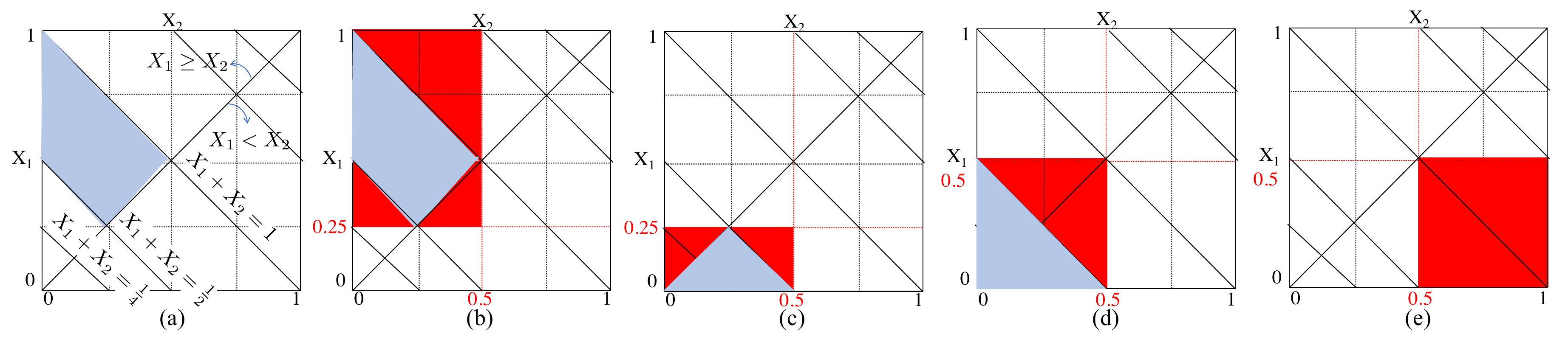}
\caption{\small{(a) Computing a convex region. (b)-(d) The source hyperplanes are emphasized (red) to represent the regions corresponding to several different function outcomes. Each outcome is the intersection of the tessellation with the red-shaded region. (e) A decoding error occurs.}}\label{fig:convexregionfunction_different_outcomes}
\end{figure*}

In Fig. \ref{fig:convexregionfunction_different_outcomes}-(a), we sketch how we compute a convex region via hyper binning for a simple example. An outcome, e.g., (b)-(d), is the intersection of the hyperplane tessellation formed by solid black lines with the red-shaded region specified by the sources. Some partitions, e.g., as shown in Fig. \ref{fig:convexregionfunction_different_outcomes}-(e), do not define a unique convex bin, i.e., a function outcome, causing decoding errors. Such events should have a low probability of occurrence via accurately capturing $\{{\bf a}_j,\,b_j\}_{j=1}^J$ (Props. \ref{entropy_based_distortion_necessary_condition}-\ref{MSE_sufficient_condition}).
\end{ex}

\section{Hyper Binning at Finite Blocklengths}
\label{compression_finite_blocklength}
For finite blocklengths, the rate limits in (\ref{rateregionSW}) do not hold. In that case, we can exploit the notion of Kolmogorov complexity $K(\mathbf{x}^n)$, i.e., the minimum description length of a string $\mathbf{x}^n$. Let ${\mathbf{X}^n}$ be i.i.d. integer-valued variables with entropy $H(X)$, where $\mathcal{X}$ is their finite alphabet, and $\mathbb{E}\left[\frac{K({\mathbf{X}^n})}{n}\right]$ be the average shortest description length of length-$n$ sequence $\mathbf{X}^n$. Then, there is a constant $c$ such that the relation of Kolmogorov complexity and entropy for all $n$ satisfies \cite[Ch. 7.3]{cover2012elements}:
\begin{align}
\label{Kolmogorov_complexity}
H(X)\leq \mathbb{E}\left[\frac{K({\mathbf{X}^n})}{n} \right]\leq H(X)+\frac{|\mathcal{X}|\log n}{n}+\frac{c}{n}.
\end{align}
In random binning, the $J=-\log(\Delta)$ bit quantization of ${\mathbf{X}_1^n}$ has an entropy of approximately $h({\mathbf{X}_1^n})+J$, where the quantization bin length $\Delta$ satisfies $\Delta=2^{-J}$. For the $J$ bit quantization of a string ${\mathbf{x}_1^n}$, we obtain the average description length via the addition of $\frac{J}{n}$ bits on both sides of (\ref{Kolmogorov_complexity}) as 
\begin{align}
H(X_{\Delta})\leq \mathbb{E}\Big[\frac{K(\{X_{\Delta}(l)\}_{l=1}^n)}{n} \Big]
\leq H(X_{\Delta})+\frac{|\mathcal{X}_{\Delta}|\log n}{n}+\frac{c}{n},\nonumber
\end{align}
where $\mathcal{X}_{\Delta}$ is the alphabet for the quantized variable $X_{\Delta}$ with $|\mathcal{X}_{\Delta}|=2^J$, and $H(X_{\Delta})\approx \frac{1}{n}h({\mathbf{x}_1^n}) + \frac{J}{n}$ bits. The finite length $n$ description of $J$-bit quantization of ${\mathbf{X}_i^n}$, for $i\in\{1,2\}$ requires an additional $\frac{|\mathcal{X}_{\Delta}|\log n}{n}$ bits on top of quantization.

From (\ref{nmax}), we have $n\leq \frac{J}{2}+O\left(\frac{1}{J}\right)$. Combining this with (\ref{Kolmogorov_complexity}), the representation complexity of random binning due to the separation of quantization and compression phases is approximately $2$ bits higher than that of hyper binning. 
Hyper binning, unlike  orthogonal binning, eliminates the need for post-quantization. The $J$ bit vector quantization is tailored for the functions, and each function outcome relies on a collection of binary decisions. This process does not involve the quantization of continuous variables, i.e., approximating the differential entropy via the addition of $J$ bits. The complexity is solely determined based on the binary entropy function.

We show the diagram of hyper binning for compression in Fig. \ref{fig:hyperbinning}. Sampling is a suboptimal single-letter approach. In information theory,  coding and compression typically follow signal processing. Hyper binning does the compression step after signal processing and before coding. It captures the entire data vector instead of a single-letter representation, giving a functional equivalence of vector quantization.

\section{Comparisons with the Existing Work}
\label{comparison_existing_work}
Characteristic hypergraph coloring in \cite{basu2020hypergraph} relies on $\epsilon-$ achievable schemes where the hyperedge-based construction exploits the fine granularity properties in the graph when a non-zero distortion is allowed. This distortion notion is more unified that generalizes the characteristic graph coloring approach in \cite{feizi2014network} via the modeling of hyperedges.

To contrast the hyper binning scheme with the existing work on graph coloring in \cite{feizi2014network} and its hypergraph-based coloring extension in \cite{basu2020hypergraph}, we next consider the Example 1 in \cite{basu2020hypergraph}. 

{\em $\epsilon$-characteristic hypergraphs vs $D$-characteristic graphs.} The $D$-characteristic hypergraph of $X$, $G^D_X$, has the vertex set $\mathcal{X}$ with any set $S \subseteq \mathcal{X}$ forming a hyperedge in $G^D_X$ if for any $x_1,\, x_2 \in S$, $d(x_1, x_2) \leq D$ where $d(\cdot)$ is some metric on $\mathcal{X}$. For independent sources $X_1$ and $X_2$, an outer bound to the achievable rate region using $D$-characteristic hypergraphs is given by \cite[Thm 43]{feizi2014network} as $\mathcal{R}^{(D/2)}_{G_{X_1,X_2}} = (R_1, R_2)$ such that
\begin{align}
R_i\geq H_{G_{X_i}(D/2)}(X_i),\quad i\in\{1,2\},    
\end{align}
where $H_{G_X(D)}(X)=\min\nolimits_{X\in W\in\Gamma(G_X^D)} I(W;X)$, and $W$ and $\Gamma(G_X^D)$ denote a hyperedge and the set of hyperedges in $G_X^D$.

\begin{ex}
Let $(X_1 , X_2 )$ be i.i.d. Bern($1/2$) variables. The decoder wants to compute the identity function $f(x_1, x_2 ) = (x_1 , x_2 )$ with $\epsilon = 0.5$. Authors in \cite{basu2020hypergraph} have demonstrated that this example achieves equality in the Berger-Tung bound $(R_1,R_2)\in\mathcal{R}_{i,\epsilon}$ and is optimal. The optimal rate region satisfies $(R_1,R_2)=\mathcal{R}_{\mathcal{G},\epsilon}$ such that $R_1\geq 0$, $R_2\geq 0$, and $R_1+R_2\geq 1$. For the same setting, using the approach in \cite{feizi2014network}, where $(R_1,R_2)= \mathcal{R}_{G_{X_1},G_{X_2}}^{D/2}$, the achievable rate is $R_i\geq 1$, which is contained in the inner region devised in \cite{basu2020hypergraph}. In hyper binning, since the function $f$ is identity and $X_1\independent X_2$, i.e., no CI between the sources, the rate region problem becomes equivalent to that of Slepian-Wolf with a distortion metric. For a fair comparison, we need to make a connection between entropy-based distortion, e.g., in \cite{courtade2011multiterminal}, versus the notion of $\epsilon$-characteristic graphs. Exploiting \cite{courtade2011multiterminal}, the rate region satisfies $R_1\geq H(X_1)-\epsilon\delta$, $R_2\geq H(X_2)-\epsilon(1-\delta)$ for $\delta\in [0,1]$. Since $H(X_i)=1$, $R_1\geq 1$ and $R_2\geq 0$ (and similarly $R_1\geq 0$ and $R_2\geq 1$) are achievable. Due to time-sharing, hyper binning can satisfy the optimal rate region of Berger-Tung. Exploiting the Hamming distortion where $\mathbb{P}(X_1\neq \hat{X}_1)\leq \epsilon$, the rate-distortion function for $X_1\sim$Bern($0.5$) satisfies 
\begin{align}
R_1(\epsilon)=(1-h(\epsilon))\cdot \mathbbm{1}_{0\leq \epsilon\leq 0.5}.
\end{align}
For recovering $(X_1,\,X_2)$ under the maximum norm constraint, letting $\norm{X_i-\hat{X}_i}\leq \epsilon_i$ for $i\in\{1,2\}$, we have $\sum\nolimits_{i=1}^2 (X_i-\hat{X}_i)^2\leq \sum\nolimits_{i=1}^2 \epsilon_i^2\leq \epsilon^2$. In this case, we obtain $R_1\geq 1,\,\, R_2\geq 1$ if $0\leq\epsilon_1,\,\epsilon_2< 1$, which implies $\mathbb{P}(X_i\neq \hat{X}_i)\leq 0$, and $R_1\geq 0,\,\, R_2\geq 0$ if $\epsilon_1,\,\epsilon_2\geq 1$, implying $\mathbb{P}(X_i\neq \hat{X}_i)\leq 1$.

Depending on the distortion criterion, we can achieve the same rates, e.g., for entropy-based distortion, as \cite{basu2020hypergraph}, or higher rates, e.g., for Hamming distortion. This conclusion holds as maximal distortion is, in general, restricted to discrete sources. It does not generalize to continuous variables, especially when we do not exploit the hypergraph structure. 
\end{ex}

We next consider a numerical example where there is no side information, which is in line with Example 2 in \cite{basu2020hypergraph}.
\begin{ex}
Let $X$ be uniformly distributed over $\{0,1,2\}$ and $f(X)=X$. Authors in \cite{basu2020hypergraph} have shown $R\in \mathcal{R}_{\mathcal{G},\epsilon}$ such that 
\begin{align}
R\geq\min\limits_{X\in W\in \Gamma(G_X^{\epsilon})} I(X;W)=\begin{cases}\log_2(3),\quad &0\leq\epsilon<0.5,\\
2/3,\quad &0.5\leq \epsilon<1,\\
0,\quad &1\leq\epsilon,\end{cases}\nonumber
\end{align}
where $G_X^{\epsilon}$ is an $\epsilon$-achievable hypergraph such that $\mathbb{E}[\mathbbm{1}_{\norm{W-X}>\epsilon}]=0$. If $\epsilon\in [0.5,\,1)$, then $H(W)=1$ since there are two maximal independent sets with $0.5$ probability each. Furthermore, $H(W|X)=\frac{1}{3}$ because $H(W|X=1)=1$ that happens with probability $\frac{1}{3}$ and $H(W|X\neq 1)=0$. Exploiting $D$-characteristic graph compression (no hyperedges) in \cite{feizi2014network}, the rate region specified by $R\in\mathcal{R}_{G_X}^D$ is given as
\begin{align}
R\geq \min\limits_{X\in W\in \Gamma(G_X^D)} I(X;W)= \begin{cases}
\log_2(3),\quad &0\leq D<2,\\
0,\quad &2\leq D.
\end{cases}\nonumber
\end{align}
In \cite{feizi2014network}, different from \cite{basu2020hypergraph}, when $D\in [1,\,2)$, the independent sets are singletons because there is no notion of hyperedges, and all source outcomes need to be distinguished. However, for $2\leq D$, we no longer need to differentiate the outcomes.
\end{ex}

In \cite{basu2020functional}, the authors extended the coloring scheme in \cite{feizi2014network} via hypergraphs. The graph $G_{X_i}^{\epsilon}$, $i\in\{1,2\}$ is an $\epsilon$-achievable hypergraph such that  $\mathbb{E}[\mathbbm{1}_{\norm{f(X_1,X_2)-\hat{f}(X_1,X_2)}>\epsilon}]=0$. The scheme in \cite{basu2020functional} results in a smoother decay in rate-distortion than that of \cite{feizi2014network}. Since the approaches in \cite{feizi2014network} and \cite{basu2020hypergraph} are for compressing  
post-quantized variables, without optimizing the quantization phase, for a fair comparison of hyper binning with them, we next draw an example with continuous variables.

\begin{ex}
Let $X_1$ and $X_2$ be distributed according to standard normal distribution $\mathcal{N}(0,1)$ and consider the function in (\ref{hyperplane_example_numbered}). Letting $X_{i,\Delta}=\Delta l$, for $X_i\in [l\Delta,\, (l+1)\Delta)$ and $i\in\{1,2\}$, and using the CDF of the standard normal distribution, denoted by $\Phi(x)=\frac{1}{\sqrt{2\pi}}\int\nolimits_{-\infty}^x e^{-t^2/2}{\rm d}t$, the quantized variables satisfy $\mathbb{P}(X_{1,\Delta}=\Delta l)=\Phi((l+1)\Delta)-\Phi(l\Delta)$. Evaluating (\ref{hyperplane_example_numbered}) using the quantized variables, we get $\mathbb{P}\Big(X_{1,\Delta}>\frac{b_1}{a_1}\Big)=\sum\nolimits_{\{l:\,l>\frac{b_1}{a_1 \Delta}\}} [\Phi((l+1)\Delta)-\Phi(l\Delta)]$, and similarly for $X_{2,\Delta}$. Then source $i\in\{1,2\}$ needs to decide whether $a_i x_i>b_i$ or not. The rate required for this model is
\begin{align}
\label{cont_example_Feizi}
h\Big(\mathbb{P}\Big(X_{1,\Delta}>\frac{b_1}{a_1}\Big)\Big)+h\Big(\mathbb{P}\Big(X_{2,\Delta}>\frac{b_2}{a_2}\Big)\Big).    
\end{align}
To achieve compression with the desired distortion (quantizer bin length $\Delta=2^{-J}$), this approach requires $J$ bits per source.

For this example, the encoding rates for different $\epsilon$ are:
\begin{itemize}[leftmargin=*]
    \item If $\epsilon\in [0,\,1)$, then for function in (\ref{hyperplane_example_numbered}) the result is identical to that of \cite{feizi2014network} and \cite{basu2020hypergraph}. If the source distributions are uniform, e.g., Gaussian variables binary quantized around the means, each user needs $1$ bit for compression, i.e., $\mathcal{R}_{G_{X_1},G_{X_2}}^{\epsilon}=\mathcal{R}_{\mathcal{G},\epsilon}$ such that $R_i\geq 1$ for $i\in\{1,2\}$. 
    \item If $\epsilon\in [1,\,2)$, the set of independent sets are 
    $\{\{1,2\},$ $\{2,3\},\{3,4\}\}$. 
    Given the interval of $X_2$, $X_1$ yields either of the hypergraphs $\{\{1,2\},\{2,3\}\}$ or $\{\{2,3\},\{3,4\}\}$. If the sources are uniform, each of these graphs has entropy $H(W)=1$, and $H(W|X_1)=1/4$. In either case, it holds that $\mathcal{R}_{\mathcal{G},\epsilon}=(R_1,R_2)$, where $R_1\geq 1$, $R_2\geq 3/4$, and the sum rate is $1+3/4=7/4$ in \cite{basu2020functional}. Similarly, $\mathcal{R}_{G_{X_1},G_{X_2}}^{\epsilon}=(R_1,R_2)$ where $R_1=R_2\geq 1$ in \cite{feizi2014network} since given $X_2$, $X_1$ yields either 
    $\{1,3\}$ or $\{2,4\}$. 
    If $\epsilon\in[1,\,2)$, given the interval of $X_1$, $X_2$ yields either $\{\{1,2\}\}$ or $\{\{3,4\}\}$. In either case, the sum rate is $1+0=1$ in \cite{basu2020functional}. In \cite{feizi2014network} $\mathcal{R}_{G_{X_1},G_{X_2}}^{\epsilon}$, where $R_1=R_2\geq 1$, and the sum rate is $1+1$ since given $X_1$, $X_2$ yields either $\{1,2\}$ or $\{3,4\}$. 
    
    \item If $\epsilon\in [2,\,3)$, $\mathcal{R}_{\mathcal{G},\epsilon}$ is such that $R_1\geq 1$, $R_2\geq 0$, and the sum rate is $1+0=1$ in \cite{basu2020functional}, which is similarly as in \cite{feizi2014network}. 
    
    In functional compression of (\ref{hyperplane_example_numbered}) the chain rule does not  
    hold \cite{feizi2014network}. To keep the sum rate constant if we swap $X_1$ and $X_2$, the distortion $\epsilon$ can be scaled by $1/2$. This is because given $X_1$, the function outcome lies either in $\{1,2\}$ or $\{3,4\}$, i.e., $R_2\geq 0$ if $\epsilon>1$. If instead $X_2$ is given, the outcome lies
    either in $\{1,3\}$ or $\{2,4\}$, i.e., $R_1\geq 0$ if $\epsilon>2$. 
\end{itemize}
The weak law of large numbers (WLLN) states that the sample average ${\overline {X}}_{n}=\frac {1}{n}\sum\nolimits_{l=1}^n{X(l)}$ converges in probability towards the expected value, i.e., ${\overline {X}}_{n}\to \mu \,\, {\textrm {as}}\ n\to \infty$. Hence, in hyper binning while for single letter representation it holds that $\mathbb{P}(X_1>\frac{b_1}{a_1})=1-\Phi\left(\frac{b_1}{a_1}\right)$, we observe that $\mathbb{P}({\overline {X_1}}_{n}>\frac{b_1}{a_1})\to \{0,1\}$ as $n\to \infty$. As a result, compressing the length-$n$ source vector provides a more accurate compression. The WLLN is true even if the summands are independent but not identically distributed \cite{Wolpert2016}. For large blocklengths, the rate for the single letter representation of hyper binning is
\begin{align}
\label{cont_example_hyper_binning_unit_length}
h\Big(\mathbb{P}\Big(X_1>\frac{b_1}{a_1}\Big)\Big)+h\Big(\mathbb{P}\Big(X_2>\frac{b_2}{a_2}\Big)\Big).    
\end{align} 
Provided that the sources are uniformly distributed about the planes, we have that $\mathbb{P}\left(X_1>\frac{b_1}{a_1}\right)=\mathbb{P}\left(X_2>\frac{b_2}{a_2}\right)=\frac{1}{2}$, and the sum rate satisfies $1+1=2$. However, this rate is clearly not achievable for finite blocklengths. In the non-asymptotic regime, exploiting the Kolmogorov complexity we can characterize the performance \cite[Ch. 7.3]{cover2012elements}.

In the asymptotic blocklength regime, using $J$ hyperplanes where $J$  properly scales with $n$, and $\{{\bf a}_{ij},\,b_{ij}\}_{j=1}^J$ for sources $i\in\{1,2\}$, the average coding rate for hyper binning is 
\begin{align}
\label{cont_example_hyper_binning_n_length}
\frac{1}{n} \sum\limits_{j=1}^J h\left(\mathbb{P}({\bf a}_{1j}^{\intercal}{\mathbf{X}_1^n}>b_{1j})\right)+
\frac{1}{n} \sum\limits_{j=1}^J
h\left(\mathbb{P}({\bf a}_{2j}^{\intercal}{\mathbf{X}_2^n}>b_{2j})\right)\nonumber\\
\leq \frac{J}{n}\sum\limits_{i=1}^2 h\Big(\frac{1}{J}\sum\limits_{j=1}^J\mathbb{P}({\bf a}_{ij}^{\intercal}{\mathbf{X}_i^n}>b_{ij})\Big),
\end{align}
where the inequality in (\ref{cont_example_hyper_binning_n_length}) follows from the concavity of entropy. The result of $\frac{1}{J}\sum\nolimits_{j=1}^J\mathbb{P}({\bf a}_{ij}^{\intercal}{\mathbf {X}_i^n}>b_{ij})$ is a probability. 

The classical orthogonal binning, i.e., random binning, is such that each sequence is uniformly assigned to one of $2^{nR_1}$ bins where $R_1>H(X_1)$ and bin $\mathcal{B}(m)$ denotes the subset of sequences with the same index $m=1,\dots,2^{nR_1}$. Evaluating the probability $\mathbb{P}({\bf a}_{1j}^{\intercal}{\mathbf{X}_1^n}>b_{1j})$ we obtain 
\begin{align}
\mathbb{P}({\bf a}_{1j}^{\intercal}{\mathbf{X}_1^n}>b_{1j})&=\sum\limits_{\{m:\, {\bf a}_{1j}^{\intercal}{\mathbf{X}_1^n}>b_{1j}\}}\mathbb{P}({\mathbf{X}_1^n}\in\mathcal{B}(m))=\zeta_{1j},\nonumber
\end{align}
where $\mathbb{P}({\mathbf{X}_1^n}\in\mathcal{B}(m))=2^{-nR_1}$, and $\zeta_j$ for a given $j$ represents the fraction of bins such that ${\bf a}_{1j}^{\intercal}{\mathbf{X}_1^n}>b_{1j}$. Hence, 
\begin{align}
\frac{1}{J}\sum\limits_{j=1}^J\mathbb{P}({\bf a}_{1j}^{\intercal}{\mathbf{X}_1^n}>b_{1j})=\frac{1}{J}\sum\limits_{j=1}^J\zeta_{1j}=\mathbb{E}[Z_1], 
\end{align}
where $Z_1=\zeta_{1j}$ with probability $1/J$ for any $j\in\{1,\dots,J\}$. Exploiting the random binning approach, the RHS of (\ref{cont_example_hyper_binning_n_length}) is
\begin{align}
\frac{J}{n}h\left(\mathbb{E}[Z_1]\right)+\frac{J}{n}h\left(\mathbb{E}[Z_2]\right).    
\end{align}
In the case of $J=2$ hyperplanes and uniform probabilities such that $\mathbb{P}({\bf a}_{ij}^{\intercal}{\mathbf{X}_1^n}>b_{1j})=1/2$ for $i\in\{1,2\}$, this yields a sum rate of $Jh\left(\mathbb{E}[Z_1]\right)+Jh\left(\mathbb{E}[Z_2]\right)=J\cdot 1+J\cdot 1=4$ bits (ignoring the scaling with $n$). However, if the distribution is not uniform such that e.g., for each $i\in\{1,2\}$ we have $\mathbb{P}({\bf a}_{ij}^{\intercal}{\mathbf{X}_i^n}>b_{ij})=1/4$ for $j=1$ and $\mathbb{P}({\bf a}_{1j}^{\intercal}{\mathbf{X}_1^n}>b_{1j})=3/4$ for $j=2$, then the LHS of (\ref{cont_example_hyper_binning_n_length}) equals $h(1/4)+h(3/4)+h(1/4)+h(3/4)$. Hence the sum rate is $3.245$ bits, indicating the savings ($0.755$ bits in the asymptotic regime) over classical random binning. 

In the non-asymptotic regime, exploiting (\ref{Kolmogorov_complexity}) we can characterize the encoding rate more precisely.
\end{ex}

\section{A Discussion on Computational Information Theory and Comparison with Modular Schemes}
\label{discussion}

In this section, to devise a new perspective on computational information theory, we provide connections between our distributed computationally aware quantization scheme that relies on hyper binning and 
the coloring-based coding models for distributed functional compression. First, in Sect. \ref{coloring_based_coding}, we describe coloring-based modular coding models that decouple coloring from Slepian-Wolf compression. Next, in Sect. \ref{HyperBinning}, we shift our focus to describe an achievable encoding for hyper binning and detail the encoding implementation in 3 steps.

\subsection{Hyper Binning vs Coloring-based Coding Schemes}
\label{coloring_based_coding}

Since the sources cannot communicate with each other, the only way to rate reduction is through a source's defining its equivalence class for functional compression. We next give a block function example for which codebook trimming followed by the Slepian-Wolf encoding is asymptotically optimal.

\begin{ex}\label{block_function_example}
{\bf A trimmable codebook.}
Two sources $X_1\independent X_2$ 
are uniformly distributed over the alphabets $\mathcal{X}_1=\mathcal{X}_2=\{0,1,2,3\}$. The function is $f(X_1,\,X_2)=X_1\xor X_2$. Note that this function exhibits the behavior as shown in Fig. \ref{fig:MI} (a). Given the function, source $1$ can determine an equivalence class $[x_1]$ which is mapped to $f(x_1,\,X_2)$. Similarly, source $2$ can determine an equivalence class $[x_2]$ mapped to $f(X_1,\,x_2)$. For this model, $[0]=[2]$ and $[1]=[3]$ both for $X_1$ and $X_2$, i.e., each source needs $1$ bit to identify themselves since the data distributions are uniform. However, the entropy of the function is $1$ bit because there are only 2 equally likely classes. 

For computing the function, each source specifies its equivalence class without any help from the other source. To specify its equivalence class $[x_1]$ source $1$ to transmit $R_1=1$ bit. Similar arguments follow for source $2$ and $R_2=1$. Hence, $R_1+R_2=2$. In this example, each equivalence class is equiprobable and has the same size, which is $2$ for each source since the model is symmetric, making the setup more tractable.
\end{ex}

While for a specific class of functions, random binning or orthogonal trimming of the binning-based codebook work,  
we conjecture that such techniques may not optimize the rate region for general functions (even without correlations). However, as authors have shown in \cite{basu2020hypergraph} that for independent sources, the Berger-Tung inner and outer bounds converge, and hence the rate of their hypergraph-based scheme lies between the bounds of \cite{tung1978multiterminal}  
and is optimal for general functions.

For functions with particular structures, e.g., the block function shown in Fig. \ref{fig:MI} (a), we can trim the binning-based codebook, as we detailed in Example \ref{Different_binning_examples}. In general, trimming may not work, e.g., the smooth function in Fig. \ref{fig:MI} (b). We next provide an example where orthogonal binning of a codebook is suboptimal for distributed functional compression.

\begin{ex}\label{BlockCounterEx}
Let $f(X_1,\,X_2)=(X_1\cdot X_2)$ mod $2$ with discrete alphabets $\mathcal{X}_1=\{1,2,3,4\}$ and $\mathcal{X}_2=\{0,1\}$. We infer that 
\begin{align}
f=0&\Rightarrow 
\tilde{x}_1\in\mathcal{X}_1,\,\, \tilde{x}_2=0,\quad \rm{or}\quad
\hat{x}_1\in\{2,4\},\,\, \hat{x}_2=1.\nonumber\\
f=1&\Rightarrow \tilde{x}_1\in\{1,3\},\,\, \hat{x}_2=1,
\end{align}
but $f(3,1)\neq f(2,1)$. We illustrate the source pairs causing distinct outcomes in Fig. \ref{fig:BlockCounter}, indicating that trimming of orthogonal bins may not work even if sources have no correlation.

\begin{figure}[h!]
\centering
\includegraphics[width=0.35\textwidth]{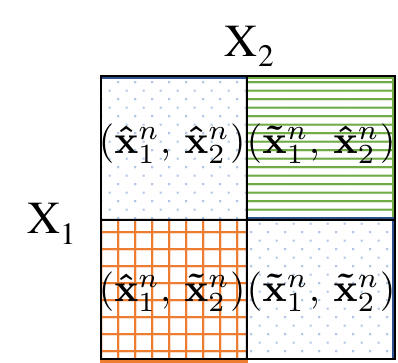}
\caption{\small{Source combinations for computing $f(X_1,\,X_2)$ in Example \ref{BlockCounterEx} for which trimming of orthogonal codebook does not hold. We fill in the source pairs causing different outputs with different patterns.}}\label{fig:BlockCounter}
\end{figure}
\end{ex}

From Example \ref{BlockCounterEx} we conjecture that orthogonal binning is in general not 
efficient when computing general functions and / or with correlated sources. To see that when the decoder observes $f(\mathbf{\hat{x}}_1^n,\mathbf{\hat{x}}_2^n)$, it is possible that $f(\mathbf{\hat{x}}_1^n,\mathbf{\hat{x}}_2^n)=f(\mathbf{\tilde{x}}_1^n,\mathbf{\tilde{x}}_2^n)$ for some source pair $(\mathbf{\tilde{x}}_1^n,\mathbf{\tilde{x}}_2^n)\neq (\mathbf{\hat{x}}_1^n,\mathbf{\hat{x}}_2^n)$. In this case, the bins cannot be combined since $f(\mathbf{\hat{x}}_1^n,\mathbf{\tilde{x}}_2^n)\neq f(\mathbf{\tilde{x}}_1^n,\mathbf{\hat{x}}_2^n)$ in general. Hence, orthogonal binning is clearly suboptimal.

Exploiting the notion of characteristic graphs,  
the authors in \cite{AlonOrlit1996}, \cite{OrlRoc2001} have recently devised coloring-based approaches and used them in characterizing rate bounds in various functional compression setups. We use the notation $H_{G_{X_i}}(X_i)$ to represent the graph entropy for the characteristic graph $G_{X_i}$ that captures the equivalence relation source $X_i$ builds for a given function $f$ on the source random variables $(X_1,\hdots,X_s)$.

\begin{defi}\label{def19FM14} \cite[Defn. 19]{feizi2014network}
A joint-coloring family $V_C= \{v_c^1,\hdots, v_c^l\}$ for $X_i$ with any valid colorings $c_{G_{X_i}}$ for $i=1,\dots,s$ is such that each $v_c^i$, called a joint coloring class, is the set of points $(x_1^{i_1},x_2^{i_2},\hdots,x_s^{i_s})$ whose coordinates have the same color, i.e., $v_c^i =\{ (x_1^{i_1},x_2^{i_2},\hdots,x_s^{i_s}),(x_1^{l_1},x_2^{l_2},\hdots,x_s^{l_s}) : c_{G_{X_1}}(x_1^{i_1}) = c_{G_{X_1}} (x_1^{l_1}),\hdots,c_{G_{X_s}} (x_s^{i_s} ) = c_{G_{X_s}} (x_s^{l_s} )\}$, for any valid $i_1,\hdots,i_s$, and $l_1,\hdots,l_s$.  
$v_c^i$ is connected if between any two points in $v_c^i$, there exists a path that lies in $v_c^i$. 
\end{defi}

For any achievable coloring-based coding scheme, authors in \cite{DosShaMedEff2010} have provided a sufficient condition called the Zig-Zag Condition, and authors in \cite{feizi2014network} both a necessary and sufficient condition called the Coloring Connectivity Condition. These are modular schemes that decouple coloring from Slepian-Wolf compression. We next state the condition in \cite{feizi2014network}.

\begin{defi}\label{def20FM14} \cite[Defn. 20]{feizi2014network} 
Let $X_i$ be random variables with any valid colorings $c_{G_{X_i}}$ for $i=1,\dots, s$. A joint coloring class $v_c^i \in V_C$ satisfies the Coloring Connectivity Condition (CCC) when it is connected, or its disconnected parts have the same function values. Colorings $c_{G_{X_1}}, \hdots, c_{G_{X_s}}$ satisfy CCC when all joint coloring classes satisfy CCC.
\end{defi}

\begin{remark}
{\bf CCC vs orthogonal binning.} CCC ensures the conditions for orthogonal binning, i.e., codebook trimming. A coloring-based encoding that satisfies CCC is applicable to Example \ref{block_function_example}. However, it may be suboptimal for functions not allowing for trimming, see Example \ref{BlockCounterEx}. Let $\mathbf{\tilde{x}}_1^n\in \{1,3\}$ and $\mathbf{\hat{x}}_1^n\in\{2,4\}$ and $\mathbf{\tilde{x}}_2^n=0$ and $\mathbf{\hat{x}}_2^n=1$. Note that $(\mathbf{\hat{x}}_1^n,\mathbf{\hat{x}}_2^n) \sim (\mathbf{\hat{x}}_1^n,\mathbf{\tilde{x}}_2^n)$ and  $(\mathbf{\hat{x}}_1^n,\mathbf{\tilde{x}}_2^n) \sim (\mathbf{\tilde{x}}_1^n,\mathbf{\tilde{x}}_2^n) $ (CCC preserved). However, $(\mathbf{\tilde{x}}_1^n,\mathbf{\tilde{x}}_2^n)\not\sim (\mathbf{\tilde{x}}_1^n,\mathbf{\hat{x}}_2^n)$ (CCC not preserved). Hence, CCC is necessary for trimming. This function also explains the suboptimality of coloring-based coding in general. 
\end{remark}

\subsection{An Achievable Encoding Scheme for Hyper Binning-based Distributed Function Quantization} 
\label{HyperBinning}
We next provide a high-level abstraction for an achievable encoding of hyper binning with $s=2$ sources. 
For a function $f(X_1,\,X_2)$ known both at the sources and at the destination, let $\{\eta_1,\eta_2,\hdots, \eta_J\} \in \mathcal{H}^2 \subset \mathbb{R}^2$ be the hyperplane arrangement of size $J$ in GP that divides $\mathbb{R}^2$ into exactly $M=r(2,J)$ regions, and is designed to sufficiently quantize $f(X_1,\,X_2)$.  
Our goal is to predetermine the parameters $\{({\bf a}_j,\,b_j)\}_{j=1}^J$ that maximize $I(M)$. 
We assume that these parameters are known at both sources and sent to the destination only once. We also highlight that we provide a heuristic for encoding, instead of explicitly generating codebooks, as we describe next.

{\bf The G{\'a}cs-K{\"o}rner Common Information Carried via Hyperplanes.}     
To enable distributed computation for non-decomposable functions, we envision a helper-based distributed functional compression approach. 
Hyper binning requires the transmission of common randomness between the source data and across the data and its function, captured through the hyperplanes. The common information (CI) measures provide alternate ways of compression for computing when there is common randomness between two  
jointly distributed sources \cite{yu2016generalized,gacs1973common}. Among these measures, the G{\'a}cs-K{\"o}rner CI (GK-CI) has applications in the private constrained synthesis of sources and secrecy \cite{salamatian2016efficient} and is relevant here because it can be separately extracted from either marginal of $X_1$ and $X_2$ \cite{gacs1973common}. In distributed CI extraction, to the best of our knowledge, the GK-CI is the only CI that exploits the combinatorial structure of  
$p_{X_1 ,X_2}$ to decompose the sources into latent common and non-common parts that ideally form disjoint components of a bipartite graph before compression.  
More specifically, the GK-CI decomposition of $p_{X_1,X_2}$ partitions the bipartite graph representation of $p_{X_1,X_2}$ into a set $\mathcal{K}$ of a maximal number of connected components $\mathcal{D}_1,\dots,\mathcal{D}_{|\mathcal{K}|}$ where $|\mathcal{K}|$ is their cardinality. The GK-CI variable $\Kgkw$ represents the index of the connected component and equals $\Kgkw=  \underset{H(U\vert X_1)=H(U\vert X_2 )=0}{\arg\max}H(U)$, i.e., $\Kgkw$ can be separately extracted from either source \cite{gacs1973common}.  
The combinatorial structure of $p_{X_1,X_2}$, captured via $\Kgkw$, can be encoded through a helper as a proxy for establishing bipartitions $\mathcal{K}$, which can provide efficient encoding and transmission of data when joint typicality decoding is not possible \cite{salamatian2016efficient}. 
Letting $\mathbb{P}(\mathcal{D}_k) = \sum\nolimits_{x_1,\,x_2\in \mathcal{D}_k} p_{X_1,X_2}(x_1,x_2)$, the GK-CI between 
$X_1$ and $X_2$  \cite{gacs1973common} equals 
\begin{align}
\label{common}
H(\Kgkw )= -\sum\limits_{k\in \mathcal{K}} \mathbb{P}(\mathcal{D}_k)\log(\mathbb{P}(\mathcal{D}_k))\quad \mbox{bits}.
\end{align} 
In our distributed quantization setting, the helper should communicate in a prescribed order the hyperplane parameters that are $J(n+1)$ in total. The rate of CI is the rate of compressing the parameters $\{({\bf a}_j,\,b_j)\}_{j=1}^J$. While these parameters are real-valued, they have approximate floating-point representations. 
Furthermore, while they might need to be updated with $n$, from (\ref{nmax}), the update rates of $J$ and hence of the hyperplane parameters is logarithmic with respect to $n$.

{\bf Encoding.} 
In encoding each source $X_i$, $i=1,\,2$ independently determines an ordering of hyperplanes to compress $X_i$. Let these orderings be $O_{X_i}\subseteq\pi_{X_i}(\{\eta_1,\eta_2,\hdots, \eta_J\})$, where $\pi_{X_i}$ is the permutation of the hyperplane arrangement from the perspective of source $i$. Note that $\pi_{X_{i_1}}\neq \pi_{X_{i_2}}$ for $i_1\neq i_2$ because sources might build different characteristic graphs. 
Source $i$ determines an ordering $O_{X_i}$, which is from the most informative, i.e., decisive in classifying the source data, to the least such that the first bit provides the maximum reduction in the entropy of the function outcome. 

{\bf Transmission.} 
Because each source has the knowledge of $\{({\bf a}_j,\,b_j)\}_{j=1}^J$, it does the comparisons ${\bf a}_j {\bf x}_t \geq b_j$ for hyperplane $j$ and sends the binary outcomes of these comparisons. Hence, each source needs to send at most $J$ bits ($1$ bit per hyperplane) to indicate the region representing the outcome of $f$. 
There are at most $2^J$ possible configurations, i.e., codewords, among which nearly $|\mathcal{C}|_{\rm HP}=2^{\sum\nolimits_{j=1}^J h(q_j)}$ are typical. Source $i$ transmits a codeword that represents a particular ordering $\pi_{X_i}$. 
Hence, in the proposed scheme with $J$ hyperplanes, we require up to $2J$ bits to describe a function with $M=r(2,J)$ outcomes. This is unlike the Slepian-Wolf setting, where source $i$ has approximately $|\mathcal{C}|_{\rm SW}=2^{n H(X_i)}$ codewords to represent the typical sequences with blocklength $n$ as $n$ goes to infinity \cite{slepian1973noiseless}. 
Hence, an advantage of the hyper binning scheme over the scheme of Slepian-Wolf is that it can capture the growing blocklength $n$ with $J$ hyperplanes without exceeding an expected distortion. 
Note that as hyper binning captures the correlation between the sources as well as between the sources and the function, it provides a representation with a reduced codebook size $|\mathcal{C}|_{\rm HP}<|\mathcal{C}|_{\rm SW}$ for distributed functional compression. If using $J\ll n$ hyperplanes ensures that the majority of $q_j$ is in $\{0,1\}$, then the efficiency of the function representation is obvious. However, if $J$ linearly scales with $n$, since $H_{G_{X_i}}(X_i)$ is the entropy of the characteristic graph that source $i$ builds to distinguish the outcomes of $f$ \cite{korner1973coding}, a sufficient condition for $\sum\nolimits_{j=1}^J h(q_j)\approx n H_{G_{X_i}}(X_i)$ is that $h(q_j)\approx \frac{n}{J}H_{G_{X_i}}(X_i)$, $\forall$ $j$.

{\bf Reception.} 
At the destination, each codeword pair received from the sources yields a distinct function output that can be determined by the specific order of the received bits in the codebooks designed for evaluating the outcome of $f$ along with the CI carried via the hyperplanes.

{\bf Discussion.} Sects. \ref{coloring_based_coding} and \ref{HyperBinning} focus on achievable schemes and are suboptimal in some cases. However,  hyper binning is not modular, unlike the coloring-based approaches, e.g., graph coloring followed by Slepian-Wolf compression in  \cite{feizi2014network} or its hypergraph-based extension in \cite{basu2020hypergraph}. Hyper binning does not involve a coloring step or a separate quantization phase prior to  
compression.  
Instead, it jointly performs 
quantization and compression.  
This joint design is possible through the knowledge of the hyperplane parameters at the source sites.

\section{Conclusions}
\label{conclusion}
We introduced a distributed function-aware quantization scheme for distributed functional compression called hyper binning. While distributed source compression algorithms in general focus on quantizing continuous variables and then compressing them, hyper binning does the compression step on the functional representation, providing a natural generalization of orthogonal binning to computation. 
Optimizing the tradeoff between the number of hyperplanes and the blocklength is crucial in exploiting the high dimensional data, especially in a finite blocklength setting. The proposed model can adapt to the changes and learn from data by successively fine-tuning the hyperplane parameters with the growing data size. Due to Kolmogorov complexity, for finite blocklengths, hyper binning can be iteratively refined to capture the function accurately at a lower cost than random binning. We believe that our approach provides a fresh perspective to vector quantization for computing. 
However, we do not claim optimality. This caveat is due to the difficulty of the NP-completeness of graph entropy and practical implementation because there is no constructive algorithm. Our future work includes sampling and vector quantization for function computation from an information-theoretic standpoint. 
Extensions also include analyzing general convex bodies formed by nonlinear hyperplanes, hypersurfaces, and multivariate functions.


\section*{Acknowledgment}
Authors gratefully acknowledge the constructive feedback from Dr. Cohen, Dr. Salamatian and the anonymous reviewers.

\begin{spacing}{1}
\bibliographystyle{IEEEtran}
\bibliography{references}
\end{spacing}

\newpage
\setcounter{page}{1}
\section*{Supplementary Material}

We recall the following notation being used in the paper: $\bar{h}_M=\frac{1}{M}\sum\nolimits_{k=1}^{M}  h(p_k)$ and $\bar{h}_{M+1}=\frac{1}{M+1}\sum\nolimits_{k=1}^{M+1}h(p_k)$, and hence $(M+1)\bar{h}_{M+1}=M\bar{h}_M+h(p_{M+1})$.

Given an MSE distortion criterion $\mathbb{E}[d({\bf f^n},{\bf \hat{f}^n})]=\frac{1}{n}\sum\nolimits_{l=1}^n (f(l) - \hat{f}(l))^2\leq\epsilon$, the approximation using $J$ hyperplanes yields an 
MSE:
\begin{align}
\label{MSE_hyperplanes_samev2}
\mathbb{E}\big[\big(\sum\limits_{j=J}^{2n}{{c_j}_{\{{\bf a}_j^{\intercal} {\bf x}_t \geq b_j\}}-{d_j}_{\{{\bf a}_j^{\intercal} {\bf x}_t<b_j\}}}\big)^2\big]\leq\epsilon.
\end{align}

\subsection{Proof of Proposition \ref{MI_lowerbound}}
\label{App:MI_lowerbound}
Adding one more hyperplane, $p_k$'s decay and for given $N$, $n_k$'s also decrease but since the source data is preserved, we have $\sum\nolimits_{k=1}^M n_k=\sum\nolimits_{k=1}^{M+1} \tilde{n}_k=\sum\nolimits_{k=1}^{M+1} \alpha n_k=N$ where $\alpha\in[0,1]$. Letting $\bar{\alpha}=1-\alpha$, the following holds for $I(M+1)$:
\begin{align}
I(M+1)&= h\Big(\frac{1}{N}\sum\limits_{k=1}^{M+1} \tilde{n}_k p_k\Big) - \sum\limits_{k=1}^{M+1} \frac{\tilde{n}_k}{N} h(p_k)\nonumber\\
&=h\Big(\frac{1}{N}\sum\limits_{k=1}^{M} \alpha n_k p_k +\bar{\alpha}p_{M+1}\Big)-\alpha\sum\limits_{k=1}^{M} \gamma_k h(p_k)-\bar{\alpha}h(p_{M+1})\nonumber\\
&\overset{(a)}{\geq} \alpha h\Big(\sum\limits_{k=1}^M \gamma_k p_k\Big)+\bar{\alpha}h(p_{M+1})-\bar{\alpha}h(p_{M+1})-\sum\limits_{k=1}^{M} \alpha\gamma_k h(p_k),\nonumber
\end{align}
where $(a)$ is due to the concavity of $h$. The RHS of $(a)$ is $\alpha I(M)$. 
For asymmetric data distribution, we have $n_k = \beta p_k$. The following relation confirms the monotonicity of $I(M)$:
\begin{align}
 I(M+1)  \geq \alpha h\Big(\frac{\beta}{N}\sum\limits_{k=1}^{M}  p_k^2 \Big)- \sum\limits_{k=1}^M \frac{\alpha\beta p_k}{N} h(p_k)
 \overset{(b)}{=}\alpha I(M),\nonumber
\end{align}
where $(b)$ follows from the definition of $I(M)$ that yields $I(M)= h\Big(\frac{\beta}{N}\sum\nolimits_{k=1}^M  p_k^2\Big) - \sum\nolimits_{k=1}^M \frac{\beta p_k}{N} h(p_k)$. As $I(M+1)\geq \alpha I(M)$ where $\alpha\in [0,1]$, the final result can be obtained.

\subsection{Proof of Proposition \ref{Bounding_MI_symmetric}}
\label{App:Bounding_MI_symmetric}
The mutual information $I(M+1)$ satisfies the relation:
\begin{align}
I(M+1)
&=h\Big(\frac{M\bar{p}_M+p_{M+1}}{M+1}\Big) - \bar{h}_{M+1}\nonumber\\
&\geq \frac{M}{M+1}h(\bar{p}_M)+ \frac{1}{M+1}h(p_{M+1})-\frac{M\bar{h}_M+h(p_{M+1})}{M+1}\nonumber\\
&=\frac{M}{M+1}\left(h(\bar{p}_M)-\bar{h}_M\right)=\frac{M}{M+1}I(M),\nonumber
\end{align}
where the inequality is due to the concavity of $h$. 

Given $M$, assume $\{p_k\}_{k=1}^M$ are fixed and in the increasing order $1/2<p_1<p_2<\hdots<p_M$ and hence $\bar{h}_M$. When we increment $M$, since $h\left(\frac{M\bar{p}_M+p_{M+1}}{M+1}\right)\leq h\left(\bar{p}_M\right)$,
\begin{align}
I(M+1)
&\leq \frac{M}{M+1}\left(h\left(\bar{p}_M\right)-\bar{h}_M\right)
+\frac{h\left(\bar{p}_M\right)-h(p_{M+1})}{M+1}\nonumber\\
&=\frac{M\, I(M)}{M+1}+\frac{h\left(\bar{p}_M\right)-h(p_{M+1})}{M+1}
=I(M)+\frac{\bar{h}_M-h(p_{M+1})}{M+1}.\nonumber
\end{align}
Combining the bounds we attain the desired result.

\subsection{Proof of Proposition \ref{MI_symmetric}}
\label{App:MI_symmetric}
For the convergence argument, from (\ref{MI_bounds_difference}), taking a sum from $M=1$ to $N-1$, we have that
\begin{align}
\sum\limits_{M=1}^{N-1} \frac{\bar{h}_M-h\left(\bar{p}_{M}\right)}{M+1}\leq I(N)-I(1) 
\leq 
\sum\limits_{M=1}^{N-1}\frac{\bar{h}_M-h(p_{M+1})}{M+1}.\nonumber
\end{align}
From the law of large numbers, 
$\bar{h}_M\to \mathbb{E}[h]=0$ as $M\to \infty$ and $h(p_M)\to 0$, i.e., the sequence $\{\frac{\bar{h}_M-h(p_{M+1})}{M+1}\}$ converges to $0$. It is indeed a Cauchy sequence. Note that a sequence $x_1, x_2, x_3, \ldots$  of real numbers is called a Cauchy sequence if, for every positive real number $\varepsilon$, there is a positive integer $N$ such that for all natural numbers $m, n > N$, $|x_m - x_n| < \varepsilon$. Hence, it is convergent.

If $n_k$'s are symmetric, $I(M)$  has the behavior, as shown in Fig. \ref{fig:MI} (d). The decay rates are low if $M$ is large. However, if $M$ is small, we expect the first term to decrease slower (concavity), yielding high mutual information. As $M$ gets larger, the decrease in the first term is sharper, and the mutual information decays, which we formally investigate next:
\begin{align}
\Delta I(M+1) 
=\frac{1}{M+1}\left(\left(\bar{h}_M-\bar{p}_M\right)-\left(h\left(p_{M+1}\right)-p_{M+1}\right)\right).\nonumber
\end{align}
Since $h(p)-p$ is decreasing in $p$ for $p\geq 1/2$, we have that $h(\bar{p}_M)-\bar{p}_M \geq h(\bar{p}_{M+1})-\bar{p}_{M+1}$. However, because entropy is concave, i.e., $h(\bar{p}_M)-\bar{p}_M\geq \bar{h}_M-\bar{p}_M$ for all $M$, this does not imply that $\bar{h}_M-\bar{p}_M>h\left(p_{M+1}\right)-p_{M+1}$ for all $M$. When $M$ is small, the gap $h(\bar{p}_M)-\bar{h}_M$ is smaller and it is possible to have $\Delta I(M+1) \geq 0$. However, when $M$ gets larger, the gap $h(\bar{p}_M)-\bar{h}_M$ is larger and $\Delta I(M+1) < 0$. 

There is a global maximum $I(M^*)$ such that  $\Delta I(M+1)\approx 0$. This is true when $h\left(p_{M+1}\right)-p_{M+1} \approx \bar{h}_M-\bar{p}_M$. The value $M^*$ is unique since as $M>M^*$, the relative increase of $p_{M+1}$ is more than $\bar{p}_M$, and the relative decrease of $h\left(p_{M+1}\right)$ with respect to $h\left(\bar{p}_M\right)$ is higher and $\bar{h}_M$ is smaller than $h\left(\bar{p}_M\right)$.

\subsection{Proof of Proposition \ref{MSE_sufficient_condition}}
\label{App:MSE_sufficient_condition}
Noting that ${\bf a}_j^{\intercal} {\mathbf{x}_1^n}$ is Gaussian distributed, and $c_j$ should be its average value such that ${\bf a}_j^{\intercal} {\mathbf{x}_1^n}\geq b_j$, and similarly for $d_j$ which is the average of ${\bf a}_j^{\intercal} {\mathbf{x}_1^n}$ such that ${\bf a}_j^{\intercal} {\mathbf{x}_1^n}<b_j$. In other words, the following relationships hold for $j\in\mathcal{J}_i,\,\, i\in\{1,2\}$:
\begin{align}
\label{c_eqn}
c_j&=\mathbb{E}[{\bf a}_j^{\intercal} {\mathbf{x}_i^n}\,\vert\, {\bf a}_j^{\intercal} {\mathbf{x}_i^n}\geq b_j]=\frac{\phi(b_j)}{1-\Phi(b_j)},\\
\label{d_eqn}
d_j&=\mathbb{E}[{\bf a}_j^{\intercal} {\mathbf{x}_i^n}\,\vert\, {\bf a}_j^{\intercal} {\mathbf{x}_i^n}<b_j]=\frac{\phi(b_j)}{\Phi(b_j)},
\end{align}
where $\phi(x)$ is the density function of the standard normal distribution and $Q(x)=1-\Phi(x)$ where the $Q$-function is the tail distribution function of the standard normal distribution. Note that the ratio of the parameters satisfy $\frac{c_j}{d_j}=\frac{\Phi(b_j)}{1-\Phi(b_j)}$.

We can evaluate the relation in (\ref{MSE_hyperplanes_samev2}) via incorporating the definition $q_j=\mathbb{P}({\bf a}_j^{\intercal} {\bf x}_t\geq b_j)$ as
\begin{align}
\label{q_j_variable}
q_j=Q\left(\frac{b_j-\mathbb{E}[{\bf a}_j^{\intercal} {\bf x}_t]}{\sqrt{{\rm Var}[{\bf a}_j^{\intercal} {\bf x}_t]}}\right)=Q\left(\frac{b_j-{\bf a}_j^{\intercal}{\bm \mu}}{\sqrt{{\bf a}_j^{\intercal}\Sigma {\bf a}_j}}\right).    
\end{align} 
For the bivariate normal distribution, the pdf of the vector $[X,Y]'$ (where $X={\bf a}_j^{\intercal} {\bf x}_t$ and $Y={\bf a}_k^{\intercal} {\bf x}_t$) satisfies

\begin{align}
{\frac {1}{2\pi \sigma _{X}\sigma _{Y}{\sqrt {1-\rho ^{2}}}}}\exp \left(-{\frac {1}{2(1-\rho ^{2})}}\left[\left({\frac {x-\mu _{X}}{\sigma _{X}}}\right)^{2}\right.\right.\nonumber\\
\left.\left.-2\rho \left({\frac {x-\mu _{X}}{\sigma _{X}}}\right)\left({\frac {y-\mu _{Y}}{\sigma _{Y}}}\right)+\left({\frac {y-\mu _{Y}}{\sigma _{Y}}}\right)^{2}\right]\right),\nonumber    
\end{align}
where the means satisfy $\mu_X={\bf a}_j^{\intercal}{\bm \mu}$ and $\mu_Y={\bf a}_k^{\intercal}{\bm \mu}$, the standard derivations satisfy $\sigma_X=\sqrt{{\bf a}_j^{\intercal}\Sigma {\bf a}_j}$ and $\sigma_Y=\sqrt{{\bf a}_k^{\intercal}\Sigma {\bf a}_k}$, and the correlation is 
\begin{align}
\rho=\frac{\mathbb{E}[XY]-\mu_X\mu_Y}{\sigma_X\sigma_Y}
=\frac{{\bf a}_j^{\intercal}(\Sigma+\mu \mu^{\intercal}){\bf a}_k-\mu_X\mu_Y}{\sigma_X\sigma_Y}.\nonumber    
\end{align}

For a given pair $(j,\,k)$ such that $j\neq k$ we let
\begin{align}
\label{q_p_r}
q_{jk}=\mathbb{P}({\bf a}_j^{\intercal} {\bf x}_t\geq b_j,\,{\bf a}_k^{\intercal} {\bf x}_t\geq b_k),\nonumber\\
p_{jk}=\mathbb{P}({\bf a}_j^{\intercal} {\bf x}_t\geq b_j,\,{\bf a}_k^{\intercal} {\bf x}_t< b_k),\nonumber\\
r_{jk}=\mathbb{P}({\bf a}_j^{\intercal} {\bf x}_t< b_j,\,{\bf a}_k^{\intercal} {\bf x}_t< b_k).
\end{align}
Combining the relations (\ref{MSE_hyperplanes_samev2}), (\ref{q_j_variable}) and (\ref{q_p_r}), we obtain that 
\begin{align}
\mathbb{E}[d(f,\hat{f})]
&=\mathbb{E}\Big[\Big(\sum\limits_{j=J}^{2n}{{c_j}\mathbbm{1}_{{\bf a}_j^{\intercal} {\bf x}_t\geq b_j}-{d_j}\mathbbm{1}_{{\bf a}_j^{\intercal} {\bf x}_t< b_j}}\Big)\nonumber\\
&\cdot\Big(\sum\limits_{k=J}^{2n}{{c_k}\mathbbm{1}_{{\bf a}_k^{\intercal} {\bf x}_t\geq b_k}-{d_k}\mathbbm{1}_{{\bf a}_k^{\intercal} {\bf x}_t< b_k}}\Big)\Big]\nonumber\\
&=\sum\limits_{j=J}^{2n}\sum\limits_{k=J}^{2n} c_j c_k \mathbb{P}({\bf a}_j^{\intercal} {\bf x}_t\geq b_j,\,{\bf a}_k^{\intercal} {\bf x}_t\geq b_k)\nonumber\\
&+\sum\limits_{j=J}^{2n}\sum\limits_{k=J,k\neq j}^{2n} c_j d_k \mathbb{P}({\bf a}_j^{\intercal} {\bf x}_t\geq b_j,\,{\bf a}_k^{\intercal} {\bf x}_t< b_k)\nonumber\\
&+\sum\limits_{j=J}^{2n}\sum\limits_{k=J,k\neq j}^{2n} d_j c_k \mathbb{P}({\bf a}_j^{\intercal} {\bf x}_t< b_j,\,{\bf a}_k^{\intercal} {\bf x}_t\geq b_k)\nonumber\\
&+\sum\limits_{j=J}^{2n}\sum\limits_{k=J}^{2n} d_j d_k \mathbb{P}({\bf a}_j^{\intercal} {\bf x}_t< b_j,\,{\bf a}_k^{\intercal} {\bf x}_t< b_k).\nonumber
\end{align}
We can rewrite $\mathbb{E}[d(f,\hat{f})]$ as
\begin{align}
\mathbb{E}[d(f,\hat{f})]
&=\sum\limits_{j=J}^{2n}\sum\limits_{k=J,k\neq j}^{2n} c_j c_k \mathbb{P}({\bf a}_j^{\intercal} {\bf x}_t\geq b_j,\,{\bf a}_k^{\intercal} {\bf x}_t\geq b_k)\nonumber\\
&+\sum\limits_{j=J}^{2n}c^2_j  \mathbb{P}({\bf a}_j^{\intercal} {\bf x}_t\geq b_j)\nonumber\\
&+2\sum\limits_{j=J}^{2n}\sum\limits_{k=J,k\neq j}^{2n} c_j d_k \mathbb{P}({\bf a}_j^{\intercal} {\bf x}_t\geq b_j,\,{\bf a}_k^{\intercal} {\bf x}_t< b_k)\nonumber\\
&+\sum\limits_{j=J}^{2n}\sum\limits_{k=J,k\neq j}^{2n} d_j d_k \mathbb{P}({\bf a}_j^{\intercal} {\bf x}_t< b_j,\,{\bf a}_k^{\intercal} {\bf x}_t< b_k)\nonumber\\
&+\sum\limits_{j=J}^{2n}d^2_j  \mathbb{P}({\bf a}_j^{\intercal} {\bf x}_t< b_j)\nonumber\\
&\overset{(\ref{q_p_r})}{=}\sum\limits_{j=J}^{2n}\sum\limits_{k=J,k\neq j}^{2n} c_j c_k q_{jk}+\sum\limits_{j=J}^{2n}c^2_j  q_j
+2\sum\limits_{j=J}^{2n}\sum\limits_{k=J,k\neq j}^{2n} c_j d_k p_{jk}\nonumber\\
&+\sum\limits_{j=J}^{2n}\sum\limits_{k=J,k\neq j}^{2n} d_j d_k r_{jk}+\sum\limits_{j=J}^{2n}d^2_j  (1-q_j)\nonumber
\end{align}
\begin{align}
\label{LDA_MSE_error_bound}
&\leq \sum\limits_{j=J}^{2n}c_j q_j\sum\limits_{k=J}^{2n}  c_k +2\sum\limits_{j=J}^{2n}c_j q_j\sum\limits_{k=J,k\neq j}^{2n} d_k\nonumber\\ 
&+\sum\limits_{j=J}^{2n} d_j\sum\limits_{k=J}^{2n}  d_k (1-q_k),
\end{align}
where the last inequality follows from (\ref{q_p_r}) where we observe that $q_j=q_{jk}+p_{jk}$ and $1-q_k=p_{jk}+r_{jk}$ for any $k\neq j$.

In general, from (\ref{LDA_MSE_error_bound}) and using the definitions of $q_{jk}$, $p_{jk}$, and $r_{jk}$, it is not straightforward to determine the set of hyperplane parameters $\{{\bf a}_j,\,b_j\}$. The MSE depends on how we jointly determine ${\bf a}_j,\,b_j,\,n,\,J$. Note also that $\{c_j\}_{j=1}^J$ and $\{q_j\}_{j=1}^J$ depend on the blocklength $n$. We emphasize that it is not straightforward to derive a necessary condition for achieving the desired MSE metric.  
On the other hand, we observe that when $c_J$ is high (when the separation between two regions needs to be large) or $\epsilon$ is small, then the required number of hyperplanes, i.e., $J$, is high.

If we assume that $d_j=-c_j$, we can have the following sufficient condition to meet the MSE criterion:
\begin{align}
\sum\limits_{j=J}^{2n}c_j q_j\sum\limits_{k=J}^{2n}  c_k &-2\sum\limits_{j=J}^{2n}c_j q_j\sum\limits_{k=J,k\neq j}^{2n} c_k\nonumber\\ &+\sum\limits_{j=J}^{2n} c_j\sum\limits_{k=J}^{2n}  c_k (1-q_k)\leq \epsilon.\nonumber
\end{align}
Rearranging the above relation we obtain
\begin{align}
&-\sum\limits_{j=J}^{2n}\sum\limits_{k=J}^{2n}c_j c_k q_j + 2\sum\limits_{j=J}^{2n}c_j^2 q_j+\sum\limits_{j=J}^{2n} \sum\limits_{k=J}^{2n} c_j c_k (1-q_j)\nonumber\\
&=2\sum\limits_{j=J}^{2n}c_j^2 q_j+\sum\limits_{j=J}^{2n}c_j (1-2q_j) \sum\limits_{k=J}^{2n} c_k \nonumber\\
&=\sum\limits_{j=J}^{2n} \left(2c_j^2 q_j+c_j (1-2q_j) \sum\limits_{k=J}^{2n} c_k  \right) \leq \epsilon.\nonumber    
\end{align}
Hence, a sufficient condition to ensure the desired distortion level is given as for $j=J,\dots, 2n$
\begin{align}
\left(2c_j^2 -2c_j\sum\limits_{k=J}^{2n} c_k\right)q_j+c_j \sum\limits_{k=J}^{2n} c_k \leq \frac{\epsilon}{2n-J+1},
\end{align}
which is equivalent to the condition for any $j\in\{J,\dots, 2n\}$:
\begin{align}
\label{sufficient_condition_q_j}
q_j
\leq \frac{1}{2c_j^2 -2c_j\sum\limits_{k=J}^{2n} c_k} \cdot\Big(\frac{\epsilon}{2n-J+1}-c_j \sum\limits_{k=J}^{2n} c_k\Big),
\end{align}
where recall that from (\ref{c_eqn})  $c_j=\frac{\phi(b_j)}{1-\Phi(b_j)}$. As $b_j$ increases we expect $c_j$ to increase and $q_j$ to decrease. However, we cannot increase $b_j$ arbitrarily because $$\frac{\epsilon}{2n-J+1}-c_j \sum\limits_{k=J}^{2n} c_k>0.$$  
Rearranging this inequality and summing up both sides of the equation from $k=J$ to $k=2n$, for the sufficient condition in (\ref{sufficient_condition_q_j}) to hold it is required that $$\sum\limits_{k=J}^{2n} c_k \leq \sqrt{\epsilon}.$$

\end{document}